\documentclass[keeplastbox]{vldb}

\usepackage[ruled, noend, noline, linesnumbered]{algorithm2e}
\usepackage{booktabs}
\usepackage[rightcaption]{sidecap}
\usepackage{bold-extra}
\usepackage{flushend}
\usepackage{amsmath,amsfonts,amssymb}
\usepackage{multirow}
\usepackage{xcolor}
\usepackage{xspace}

\usepackage{paralist}

\usepackage{graphicx}
\usepackage{grffile}

\usepackage{epstopdf}

\usepackage{boxedminipage}

\usepackage{tabularx}
\usepackage{array}
\usepackage{subfig}
\usepackage{tabularx}
\usepackage{array}

\usepackage{hyperref}
\usepackage[noabbrev,capitalise]{cleveref}

\usepackage{textcase}
\usepackage{caption}

\usepackage{stfloats}

\newcommand{\meach}{\text{\textnormal{\textbf{each} }}}

\newcommand{\Sec}[1]{\hyperref[sec:#1]{\S\ref*{sec:#1}}} 
\newcommand{\Eqn}[1]{\hyperref[eqn:#1]{(\ref*{eqn:#1})}} 
\newcommand{\Fig}[1]{\hyperref[fig:#1]{Fig.\,\ref*{fig:#1}}} 
\newcommand{\Tab}[1]{\hyperref[tab:#1]{Tab.\,\ref*{tab:#1}}} 
\newcommand{\Thm}[1]{\hyperref[thm:#1]{Thm.\,\ref*{thm:#1}}} 
\newcommand{\Lem}[1]{\hyperref[lem:#1]{Lemma\,\ref*{lem:#1}}} 
\newcommand{\Prop}[1]{\hyperref[prop:#1]{Prop.~\ref*{prop:#1}}} 
\newcommand{\Cor}[1]{\hyperref[cor:#1]{Cor.~\ref*{cor:#1}}} 
\newcommand{\Def}[1]{\hyperref[def:#1]{Defn.~\ref*{def:#1}}} 
\newcommand{\Alg}[1]{\hyperref[alg:#1]{Alg.\,\ref*{alg:#1}}} 
\newcommand{\Ex}[1]{\hyperref[ex:#1]{Ex.~\ref*{ex:#1}}} 
\newcommand{\Clm}[1]{\hyperref[clm:#1]{Claim~\ref*{clm:#1}}} 
\newcommand{\Step}[1]{\hyperref[step:#1]{Step~\ref*{step:#1}}} 

\newtheorem{theorem}{Theorem}

\newtheorem{defn}{Definition}
\newtheorem{lemma}{Lemma}

\newcommand{\cC}{\mathcal{C}}
\newcommand{\cF}{\mathcal{F}}

\newcommand{\cH}{\mathcal{H}}

\newcommand{\cN}{\mathcal{N}}
\newcommand{\cO}{\mathcal{O}}
\newcommand{\cR}{\mathcal{R}}
\newcommand{\cS}{\mathcal{S}}
\newcommand{\cU}{\mathcal{U}}

\newcommand{\true}{\textsc{true}}
\newcommand{\false}{\textsc{false}}

\newcommand{\NN}{\mathbb{N}}

\newcommand{\sR}{\cR}
\newcommand{\sS}{\cS}

\newcommand{\kval}{\kappa}

\newcommand{\upf}{\cU}

\newcommand{\tassk}{\textsc{ask}\xspace}
\newcommand{\tfb}{\textsc{fb}\xspace}
\newcommand{\tslj}{\textsc{slj}\xspace}
\newcommand{\tso}{\textsc{ork}\xspace}
\newcommand{\tsse}{\textsc{sse}\xspace}
\newcommand{\tsth}{\textsc{hg}\xspace}
\newcommand{\ttw}{\textsc{tw}\xspace}
\newcommand{\twg}{\textsc{wgo}\xspace}
\newcommand{\twn}{\textsc{wnd}\xspace}
\newcommand{\twiki}{\textsc{wiki}\xspace}
\newcommand{\tfri}{\textsc{fri}\xspace}
\definecolor{applegreen}{rgb}{0.55, 0.71, 0.0}
\definecolor{asparagus}{rgb}{0.53, 0.66, 0.42}
\definecolor{brightgreen}{rgb}{0.4, 1.0, 0.0}
\definecolor{caribbeangreen}{rgb}{0.0, 0.8, 0.6}
\definecolor{chromeyellow}{rgb}{1.0, 0.65, 0.0}
\definecolor{darkolivegreen}{rgb}{0.33, 0.42, 0.18}
\definecolor{darkpastelgreen}{rgb}{0.01, 0.75, 0.24}
\newcommand{\algSND}[1]{{\color{darkpastelgreen}#1}}
\newcommand{\algAND}[1]{{\color{cyan}#1}}
\newcommand{\algSP}[1]{{\color{orange}#1}}

\definecolor{myblue}{RGB}{0,186,255}
\definecolor{myred}{RGB}{233,3,7}
\definecolor{mygreen}{RGB}{93,199,70}
\definecolor{myyellow}{RGB}{255,158,23}

\newcommand{\assk}{{\tt as-skitter}\xspace}
\newcommand{\fb}{{\tt facebook}\xspace}
\newcommand{\slj}{{\tt soc-LiveJournal}\xspace}
\newcommand{\so}{{\tt soc-orkut}\xspace}
\newcommand{\sse}{{\tt soc-sign-epinions}\xspace}
\newcommand{\sth}{{\tt soc-twitter-higgs}\xspace}
\newcommand{\tw}{{\tt twitter}\xspace}
\newcommand{\wg}{{\tt web-Google}\xspace}
\newcommand{\wn}{{\tt web-NotreDame}\xspace}
\newcommand{\wiki}{{\tt wiki-200611}\xspace}
\newcommand{\fri}{{\tt friendster}\xspace}
\newcommand{\PAND}{\textsc{PartialAnd}\xspace}
\newcommand{\AND}{\textsc{And}\xspace}
\newcommand{\SND}{\textsc{Snd}\xspace}

\newcommand{\hl}[1]{{\color{black}{#1}}}

\let\oldnl\nl
\newcommand{\nonl}{\renewcommand{\nl}{\let\nl\oldnl}}

\vldbTitle{Local Algorithms for Hierarchical Dense Subgraph Discovery}
\vldbAuthors{Ahmet Erdem Sar{\i}y\"{u}ce, C. Seshadhri, and Ali Pinar}
\vldbDOI{https://doi.org/10.14778/3275536.3275540}
\vldbVolume{12}
\vldbNumber{1}
\vldbYear{2019}

\begin{document}
\SetEndCharOfAlgoLine{}
\title{Local Algorithms for Hierarchical Dense Subgraph Discovery}

\numberofauthors{3} 
\author{
\alignauthor
Ahmet Erdem Sar{\i}y\"{u}ce\\
	\affaddr{University at Buffalo}\\
	\email{erdem@buffalo.edu}
\alignauthor
C. Seshadhri\\
	\affaddr{Univ. of California, Santa Cruz} \\
	\email{sesh@ucsc.edu}
\alignauthor
Ali Pinar\\
	\affaddr{Sandia National Laboratories} \\
	\email{apinar@sandia.gov}
}

\maketitle

\begin{abstract}
Finding the dense regions of a graph and relations among them is a fundamental problem in network analysis. Core and truss decompositions reveal dense subgraphs with hierarchical relations. The incremental nature  of algorithms for computing these decompositions  and the  need for global information at each step of the algorithm hinders  scalable parallelization and approximations since the densest regions are not revealed until the end. In a previous work, Lu et al. proposed to iteratively compute the $h$-indices of neighbor vertex degrees to obtain the core numbers and prove that the convergence is obtained after a finite number of iterations.  This work generalizes the iterative $h$-index computation for truss decomposition as well as nucleus decomposition which leverages higher-order structures to generalize  core and truss decompositions. In addition, we prove convergence bounds on the number of iterations. We present a framework of local algorithms to obtain the core, truss, and nucleus decompositions. Our algorithms are local, parallel, offer high scalability, and enable approximations to explore time and quality trade-offs. Our shared-memory implementation verifies the efficiency, scalability, and effectiveness of our local algorithms on real-world networks.
\end{abstract}


\section{Introduction}\label{sec:intro}
\noindent
A characteristic feature of the real-world graphs is  sparsity at the global level yet density in the local neighborhoods~\cite{Gleich12}.
Dense subgraphs are indicators for functional units or unusual behaviors.
They have been adopted  in various applications, such as detecting  DNA motifs in biological networks~\cite{Fratkin06}, identifying the news stories from microblogging streams in real-time~\cite{Angel12}, finding  price value motifs in financial networks~\cite{Du09}, and locating  spam link farms in web~\cite{Kumar99,Gibson05,Dourisboure07}.
Dense regions are also used to improve efficiency of compute-heavy tasks like distance query computation~\cite{Jin09} and materialized per-user view creation~\cite{Gionis13}.

Detecting dense structures in various granularities and finding the hierarchical relations among them is a fundamental problem in graph mining. For instance,  in a  citation network,  the hierarchical relations of dense parts in various granularities can reveal how new research areas are initiated or which research subjects became popular in time~\cite{Sariyuce-TWEB17}.
$k$-core~\cite{Seidman83,MaBe83} and $k$-truss decompositions~\cite{Saito06,Cohen08,Verma12,Zhang12} are effective ways to find many dense regions in a graph and construct a hierarchy among them.
$k$-core is based on the vertices and their degrees, whereas $k$-truss relies on the edges and their triangle counts.

\begin{figure}[!b]
\centering
\captionsetup[subfigure]{captionskip=1.3ex}
\vspace{-3ex}
\hspace{-16ex}
\subfloat[Convergence rates]{\includegraphics[height=2.65cm]{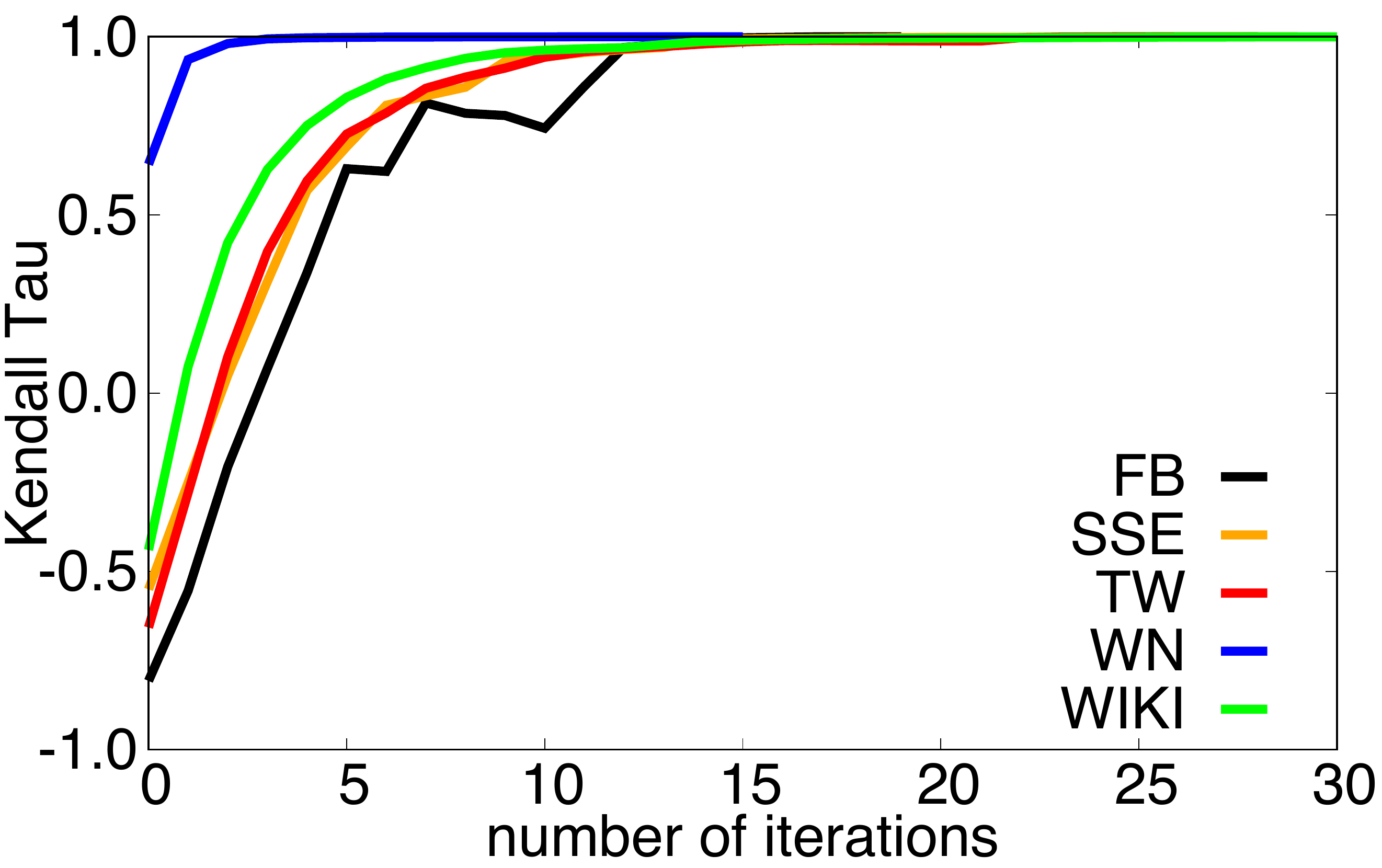}\label{fig:1a}}
\subfloat[Scalability performance]{\includegraphics[height=2.65cm]{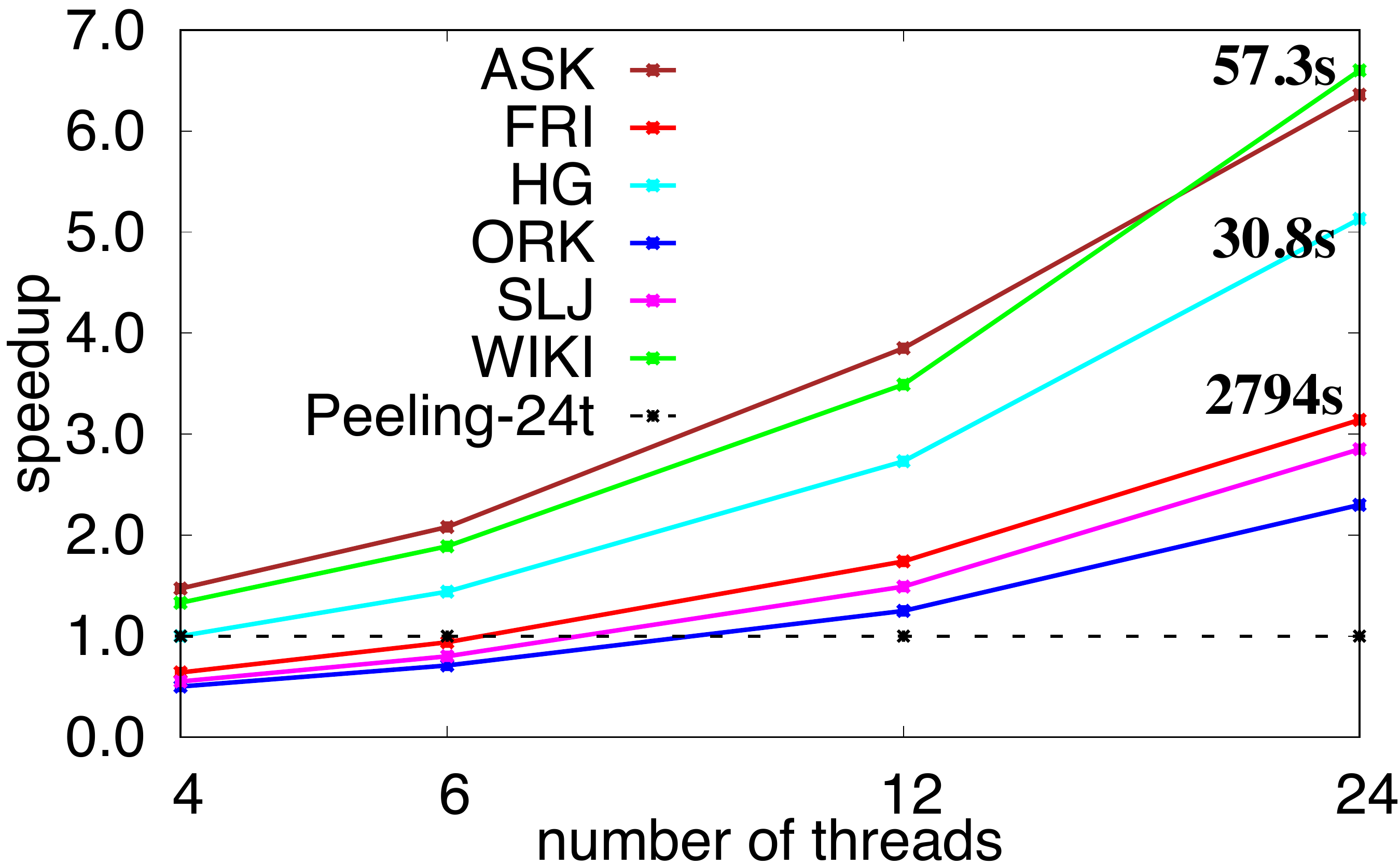}\label{fig:1b}}\vspace{-1ex}
\hspace{-18ex}
\caption{\it{On the left, we present the convergence rates for {$k$}-truss decomposition on five graphs. Kendall-Tau similarity score compares the obtained and the exact decompositions; becomes 1.0 when they are the same. Our local algorithms compute the almost-exact decompositions in around 10 iterations. On the right, we show the runtime performances w.r.t. partially parallel peeling algorithms. On average, {$k$}-truss computations are  $4.8$x faster when we switch from 4 threads to 24 threads.}}
\vspace{-3ex}
\label{fig:intro}
\end{figure}

Higher-order structures, also known as motifs or graphlets, have been used to find dense regions that cannot be detected with edge-centric methods~\cite{Benson16,Tsourakakis15}.
Computing the frequency and distribution of triangles and other small motifs  is a simple yet effective approach used in data analysis~\cite{Jha15, Pinar16, Ahmed15, Rossi17}.
Nucleus decomposition is a framework of decompositions that is able to use higher-order structures to find dense subgraphs with hierarchical relations~\cite{Sariyuce15, Sariyuce-TWEB17}.
It generalizes the $k$-core and $k$-truss approaches and finds higher-quality dense subgraphs with more detailed hierarchies.
However, existing algorithms in the nucleus decomposition framework require global graph information, which becomes a performance bottleneck for massive networks.
They are also not amenable for parallelization or approximation due to their interdependent incremental nature.
We introduce a framework of algorithms for nucleus decomposition that uses \textbf{only local information}.
Our algorithms provide  faster and approximate solutions and their local nature enables query-driven processing of  vertices/edges.

\subsection{Problem and Challenges}

\noindent The standard method to compute a $k$-core decomposition is a sequential algorithm, known as the peeling process.
To find a $k$-core, all vertices with degree less than $k$ are removed repeatedly until no such vertex remains.
This process is repeated after incrementing $k$ until no vertices remain.
Batagelj and Zaversnik introduced a bucket-based $\cO(|E|)$ algorithm for this process~\cite{BaZa03}.
It keeps track of the vertex with the minimum degree at each step, thus requires global information about the graph at any time.
$k$-truss decomposition has a similar peeling process with $\cO(|\triangle|)$ complexity~\cite{Cohen08}.
To find a $k$-truss, all edges with less than $k$ triangles are removed repeatedly and at each step, algorithm keeps track of the edge with the minimum triangle count, which requires information from all around the graph.
Nucleus decomposition~\cite{Sariyuce15} also facilitates the peeling process on the given higher-order structures.
\textbf{The computational bottleneck in the peeling process is the need for the global graph information.}
This results in inherently sequential processing.
Parallelizing the peeling process in a scalable way is challenging since each step depends on the results of the previous step.
Parallelizing each step in itself is also infeasible since synchronizations are needed to decrease the degrees of the vertices that are adjacent to multiple vertices being processed in that step.

\noindent {\bf Iterative \emph{h}-index computation:} Lu et al. introduced an alternative formulation for $k$-core decomposition~\cite{Lu16}.
They proposed iteratively computing $h$-indices on the vertex degrees to find the core numbers of vertices (even though they do not call out the correspondence of their method to h-indices). Degrees of the vertices are used as the initial core number estimates  and each vertex updates its estimate as  the $h$-index value for its neighbors' core number estimates.
This process is repeated until convergence.
At the end, each vertex has  its core number.
They prove that convergence to the core numbers is guaranteed, and analyze the convergence characteristics of the real-world networks and show runtime/quality trade-offs.

We generalize Lu et al.'s work for \emph{any} nucleus decomposition, including $k$-truss.
We show that convergence is guaranteed for all nucleus decompositions and prove the first upper bounds for the number of iterations.
Our framework of algorithms locally compute \textit{any} nucleus decomposition.
We propose that iteratively computing $h$-indices of vertices/edges/$r$-cliques based on their degrees/triangle/$s$-clique counts converges in the core/truss/nucleus numbers ($r < s$).
Local formulation also enables the parallelization.
Intermediate values  provide an approximation to the exact nucleus decomposition to trade-off between runtime and quality.
Note that this is not possible in the peeling process, because no intermediate solution can provide an overall approximation to the exact solution, e.g., the densest regions are not revealed until the end.

\subsection{Contributions}
\noindent Our contributions can be summarized as follows:
\begin{compactitem}[\leftmargin=0.1ex $\bullet$]
\item \textbf{Generalizated nucleus decomposition:} We generalize the iterative $h$-index computation idea~\cite{Lu16} for \emph{any} nucleus decomposition by using only local information.
Our approach is based on iteratively computing the $h$-indices on the degrees of vertices, triangle counts of edges, and $s$-clique counts of $r$-cliques ($r<s$) until convergence.
We prove that the iterative computation by  $h$-indices guarantees  exact core, truss, and nucleus decompositions.
\item \textbf{Upper bounds for convergence:}  We prove an upper bound for the number of iterations needed for convergence.
We define the concept of \textit{degree levels} that models the worst case for convergence.
Our bounds are applicable to  \textit{any} nucleus decomposition and much tighter than the trivial bounds that rely on the total number of vertices.
\item \textbf{Framework of parallel local algorithms:} We introduce a framework of efficient algorithms that only use local information to compute \textit{any} nucleus decomposition.
Our algorithms are highly parallel due to the local computation and are implemented in OpenMP for shared-memory architectures.
\item \textbf{Extensive evaluation on real-world networks:} We evaluate our algorithms and implementation on various real-world networks.
We investigate the convergence characteristics of our new algorithms and show that close approximations can be obtained only in a few iterations.
This enables exploring trade-offs between time and accuracy.
\cref{fig:1a} presents the convergence rates  for the $k$-truss decomposition.
In addition, we present a metric that approximates solution quality  for  informed decisions on accuracy/runtime trade-offs.
We also evaluate  runtime performances of our algorithms, present  scalability results, and examine trade-offs between runtime and accuracy. \cref{fig:1b} has the results at a glance for the $k$-truss case.
Last, but not least, we highlight a query-driven scenario where our local algorithms are used on a subset of vertices/edges to estimate the core and truss numbers.
\end{compactitem}

\section{Background}\label{sec:prelim_nuc}
\noindent We work on a simple undirected graph $G=(V, E)$ where $V$ is the set of vertices and $E$ is the set of edges.
We define $r$-clique as a complete graph among $r$ vertices for $r>0$, i.e., each vertex is connected to all the other vertices. 
\textbf{We use $R$ (and $S$) to denote $r$-clique (and $s$-clique).}

\subsection{Core, Truss, and Nucleus Decompositions}

\begin{defn}
$k$-core  of  $G$ is a maximal connected subgraph  of $G$, where each vertex has at least degree $k$.
\end{defn}

A vertex can reside in multiple $k$-cores, for different $k$ values, which results in a hierarchy.
Core number of a vertex is defined to be the the largest $k$ value for which there is a $k$-core that contains the vertex.
Maximum core of a vertex is the maximal subgraph around it that contains vertices with  equal or larger core numbers.
It can be found by a traversal that only includes the vertices with larger or equal core numbers.
\cref{fig:toycore} illustrates $k$-core examples.

\begin{defn}
$k$-truss of $G$ is a maximal connected subgraph of $G$ where each edge is in at least $k$ triangles.
\end{defn}

Cohen~\cite{Cohen08} defined the standard maximal $k$-truss as one-component subgraph such that each edge participates in at least $k-2$ triangles, but here we just assume it is $k$ triangles for the sake of simplicity.
An edge can take part in multiple $k$-trusses, for different $k$ values, and there are also hierarchical relations
between $k$-trusses in a graph.
Similar to the core number, truss number of an edge is defined to be the largest $k$ value for which there exists a $k$-truss that includes the edge and maximum truss of an edge is the maximal subgraph around it that contains edges with larger or equal truss numbers. 
We show some truss examples in \cref{fig:toytruss}.
Computing the core and truss numbers are known as the core and truss decomposition.

\begin{figure}[t!]
\centering
\captionsetup[subfigure]{captionskip=1ex, margin=5ex}
\hspace{-22ex}
\subfloat[{$k$}{-cores}]{\includegraphics[width=0.53\linewidth]{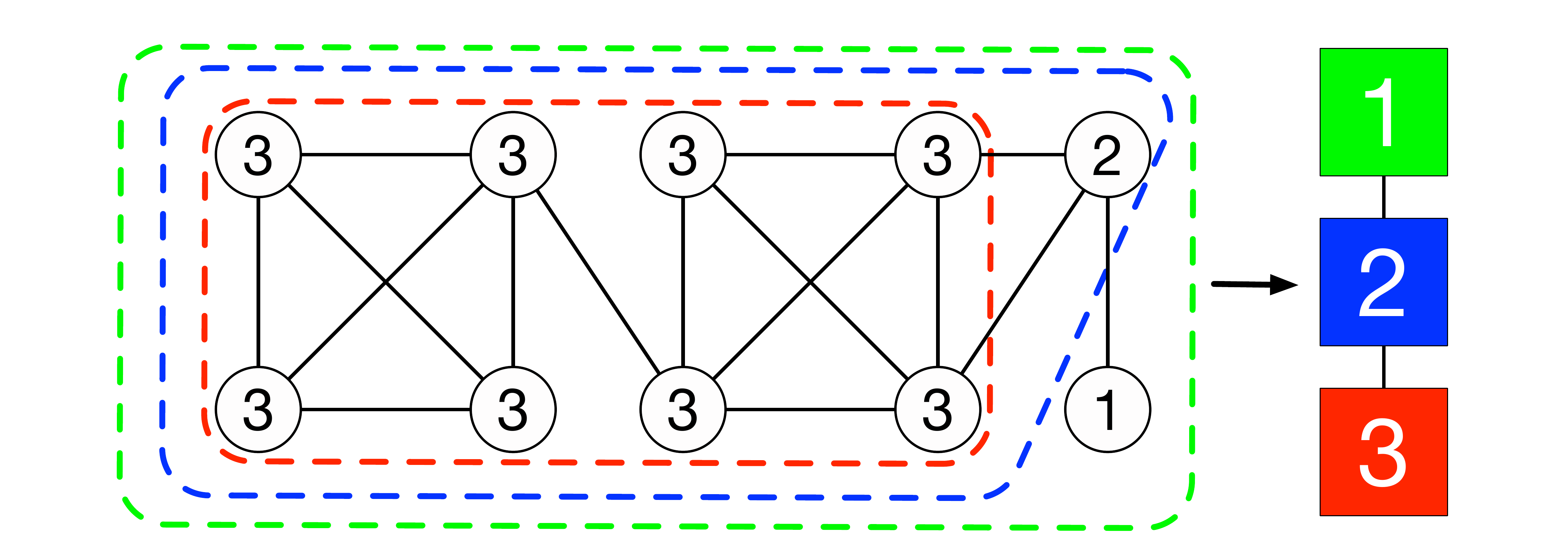}\label{fig:toycore}}
\hspace{-2ex}
\subfloat[{$k$}{-trusses}]{\includegraphics[width=0.53\linewidth]{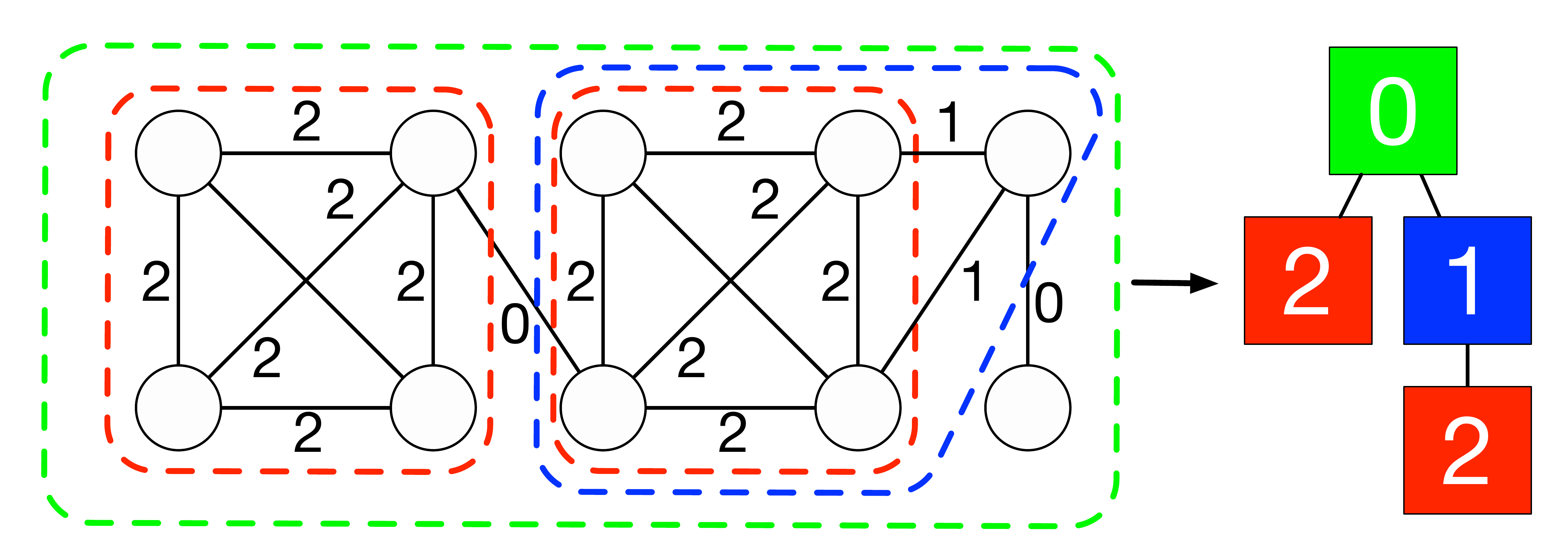}\label{fig:toytruss}}
\hspace{-19ex}
\caption{\it{Illustrative examples for $k$-core and $k$-truss. On the left, red, blue, and green regions show the 3-, 2-, and 1-cores. Core numbers are also shown for each vertex. For the same graph, trusses are presented on the right. Entire graph is a 0-truss. Five vertices on the right form a 1-truss, in blue. There are also two 2-trusses and one of them is a subset of the 1-truss. Truss numbers of the edges are also shown.}}
\vspace{-2ex}
\label{fig:toys}
\end{figure}

\begin{figure}[b!]
\vspace{-3ex}
\centering
\captionsetup[subfigure]{captionskip=0.1ex, margin=5ex}
\hspace{-21ex}
\subfloat[{$k$}{-trusses}]{\includegraphics[width=0.53\linewidth]{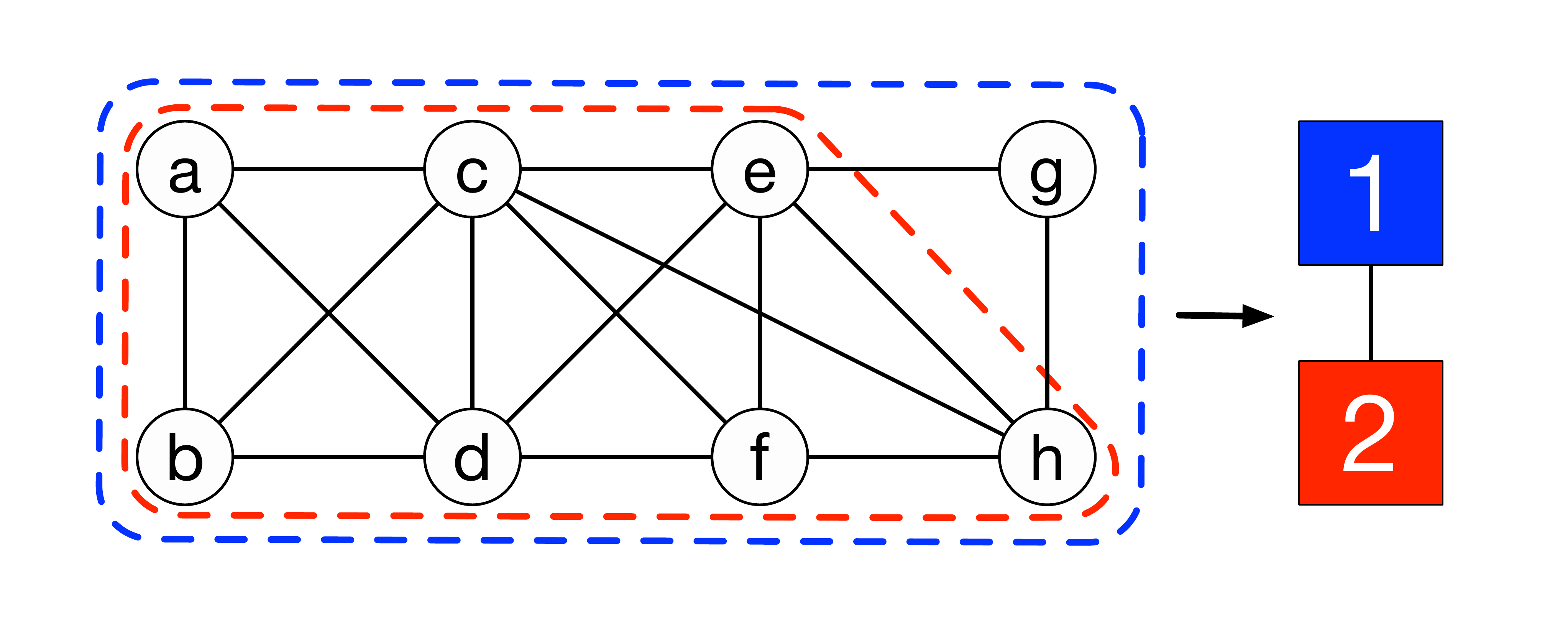}\label{fig:toytruss2}}
\hspace{-2ex}
\subfloat[{$k$}{-(3,4) nuclei}]{\includegraphics[width=0.53\linewidth]{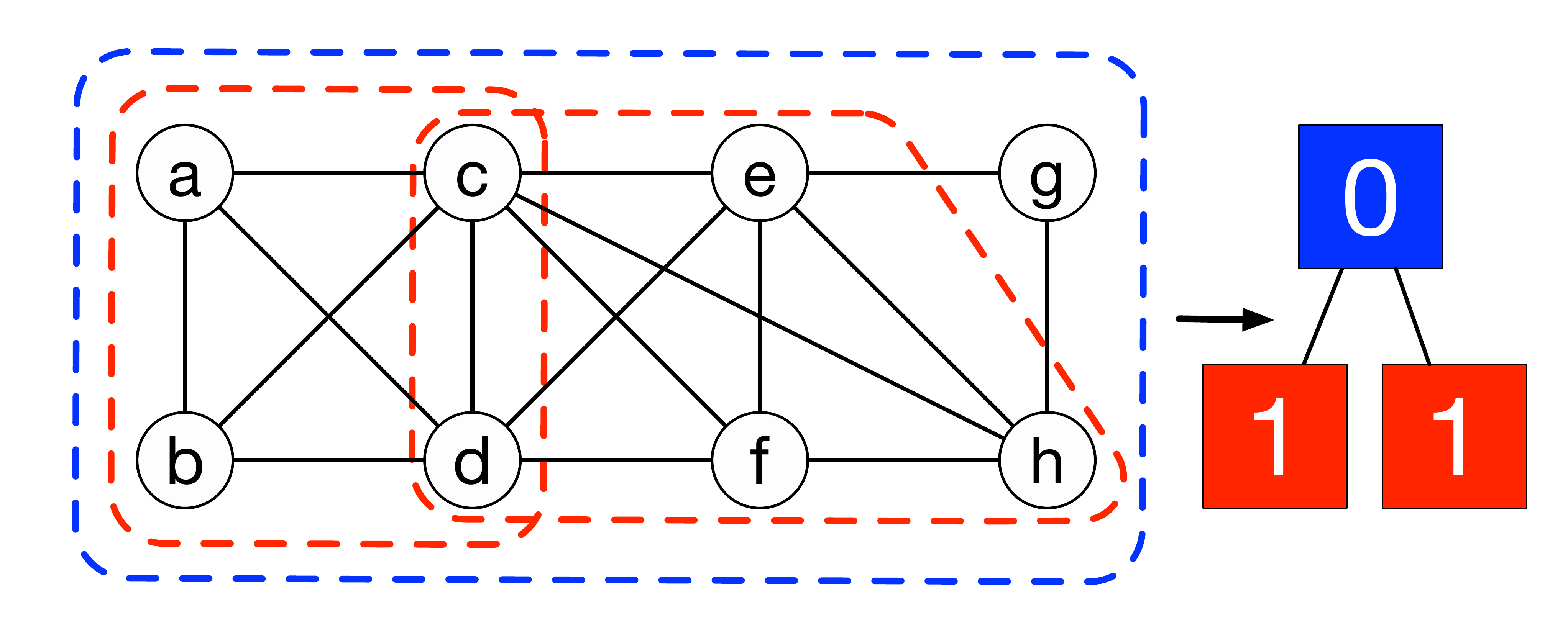}\label{fig:toynuc}}
\hspace{-20ex}
\vspace{-1ex}
\caption{\textit{Illustrative examples for $k$-truss and $k$-$(3,4)$ nucleus. On the left, entire graph is a single 1-truss and all except vertex $g$ forms a 2-truss. For the same graph, $k$-$(3,4)$ nuclei are shown on the right. Entire graph is a 0-(3,4) nucleus and there are two 1-(3,4) nuclei (in red): subgraph among vertices $a, b, c, d$ and subgraph among $c, d, e, f, h$. Note that those two subgraphs are reported as separate, not merged, since there is no four-clique that contains a triangle from each nuclei (breaking $\sS$-connectedness).}}
\vspace{0ex}
\label{fig:toynucleus}
\end{figure}

\noindent \textbf{Unifying $k$-core and $k$-truss:} Nucleus decomposition is a framework that generalizes core and truss decompositions~\cite{Sariyuce15}.
$k$-$(r,s)$ nucleus is defined as the maximal subgraph of the $r$-cliques where each $r$-clique takes part in at least $k$ $s$-cliques.
We first give some basic definitions and then formally define the $k$-$(r,s)$ nucleus subgraph.

\begin{defn} \label{def:cl-path} Let $r < s$ be positive integers.
\begin{compactitem}
\setlength\itemsep{1ex}
	\item $\sR(G)$and $\sS(G)$ are the set of $r$-cliques and $s$-cliques in $G$, respectively (or $\sR$ and $\sS$ when $G$ is  unambigous).
	\item \boldmath{$\sS$}\textbf{-degree} of $~R \in \sR(G)$ is the number of $~S \in \sS(G)$ such that $S$ contains $R$ ($R \subset S$). It is denoted as {\boldmath$d_s|_G(R)$} (or {$d_s|(R)$} when $G$ is obvious).
	\item Two $r$-cliques $R, R'$ are \emph{\boldmath{$\sS$}\textbf{-connected}} if there exists a sequence $R = R_1, R_2, \ldots, R_k = R'$ in $\sR$ such that for each $i$, some $S \in \sS$ contains $R_i \cup R_{i+1}$.
	\item Let $k$, $r$, and $s$ be positive integers such that $r < s$. A \emph{\boldmath{$k$}\textbf{-}\boldmath{$(r,s)$} \textbf{nucleus}} is a subgraph $G'$ which contains the edges in the maximal union $\cS$ of $s$-cliques such that
\begin{compactitem}
\item $\sS$-degree of any $r$-clique $R \in \sR(G')$ is at least $k$.
\item Any $r$-clique pair $R, R' \in \sR(G')$ are $\sS$-connected.
\end{compactitem}
\end{compactitem}
\end{defn}

For $r=1, s=2$, $k$-$(1, 2)$ nucleus is a maximal (induced) connected subgraph with minimum vertex degree $k$. This is exactly the $k$-core.
Setting $r=2, s=3$ gives maximal subgraphs where every edge participates in at least $k$ triangles, which corresponds to $k$-truss, and all edges are triangle connected, which is also introduced as a $k$-truss community~\cite{Huang14}.
Note that the original $k$-truss definition is different from $(2,3)$ nucleus since it does not require triangle connectivity. In this work, we focus on core, truss, and $\kappa_s$ indices (\cref{def:kappaindex}) and our algorithms work for \textit{any connectivity constraint}. We skip details for brevity.

Nucleus decomposition for $r=3$ and $s=4$ has been shown to give denser (and non-trivial) subgraphs than the $k$-cores and $k$-trusses, where density of $G=(V,E)$ is defined as $2|E|/{|V| \choose 2}$~\cite{Sariyuce15}.
\cref{fig:toynucleus} presents the difference between $k$-truss and $k$-$(3,4)$ nucleus on a toy graph.
It is used to analyze citation networks of  APS journal papers and discovered  hierarchy of topics  i.e., a large subgraph on complex networks has children subgraphs on synchronization networks, epidemic spreading and random walks, which cannot be observed with core and truss decompositions~\cite{Sariyuce-TWEB17}. 
Nucleus decomposition for larger $r$ and $s$ values are costly and only affordable for small networks with a few thousand edges. Enumerating $r$-cliques and checking their involvements in $s$-cliques for $r,s>4$, can become intractable for larger graphs, making $k$-$(3,4)$ is a sweet spot.

In graph $G$, minimum $\sS$-degree of an $r$-clique $R \in \sR(G)$ is denoted as {\boldmath$\delta_{r,s}(G)$}, e.g., the minimum degree of a vertex in $G$ is $\delta_{1,2}(G)$.
We also use {\boldmath $\cN_s(R)$} to denote the set of neighbor $r$-cliques of $R$ such that $R' \in \cN_s(R)$ if $\exists$ an $s$-clique $S$ s.t. $S\supset R$ and $S\supset R'$.
As in $k$-core and $k$-truss definitions, an $r$-clique can take part in multiple $k$-$(r,s)$ nuclei for different $k$ values.
We define the $\kval_s$ index of $r$-clique analogous to the core numbers of vertices and truss numbers of edges~\cite{Sariyuce15}.

\vspace{-2ex}
\begin{defn} \label{def:kappaindex}
For any $r$-clique $R \in \cR(G)$, the $\boldsymbol{\kval_s}$\textbf{-index} of $R$, denoted as $~\kval_s(R)$, is the largest $k$ value such that $R$ is contained in a $k$-($r,s$) nucleus.
\end{defn}
\vspace{-2ex}

\begin{table}[!t]
\vspace{1ex}
\small
\renewcommand{\tabcolsep}{2pt}
\caption{{\it Notations}}
\label{tab:notation}
\vspace{-2ex}
\hspace{-1ex}
\begin{tabular}{ | l | l |} \hline
\textbf{Symbol}	&	\textbf{Description} \\\hline\hline
$G$ & graph\\ \hline
$R$ & $r$-clique\\ \hline
$S$ & $s$-clique\\ \hline
$\cR(G)~($or~$\cR)$ & set of $r$-cliques in graph $G$\\ \hline
$\cS(G)~($or~$\cS)$ & set of $s$-cliques in graph $G$\\ \hline
$\cC(G)~($or~$\cC)$ & $\cR(G) \cup \cS(G)$\\ \hline
$d_s|_G(R)$ & $\cS$-degree of $R$: number of $s$-cliques \\
$($or~$d_s(R))$ & that contains $R$ in graph $G$ \\ \hline
$\delta_{r,s}(G)$ & minimum $\sS$-degree of an $r$-clique in graph $G$\\ \hline
$\cN_s(R)$ & neighbor $R'$s s.t. $\exists$ an $s$-clique $S \supseteq (R \cup R')$ \\ \hline
$\kval_s(R)$ & largest $k$ s.t. $R$ is contained in a $k$-$(r,s)$ nucleus \\ \hline
$\kval_2(u)$, $\kval_3(e) $ & Core number of vertex $u$, truss number of edge $e$ \\ \hline
$\cH(K)$ & largest $h$ s.t. at least $h$ numbers in set $K$ are $\ge h$\\ \hline
$\cU$ & update operator, $(\tau: \cR \to \NN) \to (\upf \tau: \cR \to \NN)$\\ \hline
\end{tabular}
\vspace{0ex}
\end{table}

\LinesNumberedHidden
\setlength{\textfloatsep}{0pt}
\begin{algorithm}[!t]
\caption{\textsc{Peeling($G, r, s$)}}
\label{alg:peeling}
\nonl \textbf{Input:} $G$: graph, $r< s$: positive integers \;
\nonl \textbf{Output:} $\kappa_s(\cdot)$: array of $\kappa_s$ indices for $r$-cliques\;
  Enumerate all $r$-cliques in $G$\;
  For every $r$-clique $R$, set $d_s(R)$ ($\cS$-degrees)\;
  Mark every $r$-clique as unprocessed\;
  \For{\meach unprocessed $r$-clique $R$ with minimum $d_s(R)$}{
  	$\kval_s(R) = d_s(R)$\;
    Find set $\cS$ of $s$-cliques containing $R$\label{ln:sk6} \;
    \For{\meach $C \in \cS$}{\label{ln:sk7}
    	\lIf{any $R \subset C$ is processed}{\label{ln:sk8}
    		continue
    	}
    	\For{\meach $r$-clique $R' \subset C$, $R' \neq R$}{
    		\lIf{$d_s(R') > d_s(R)$}{
    			$d_s(R') = d_s(R') - 1$
    		}
    	}
    }
  Mark $R$ as processed\;
  }
  \Return array $\kval_s(\cdot)$\;
\end{algorithm}

Core number of a vertex $u$ is denoted by $\kval_2(u)$ and the truss number of an edge $e$ is denoted by $\kval_3(e)$.
We use the notion of $k$-($r,s)$ nucleus and $\kval_s$-index to introduce our generic theorems and algorithms for \emph{any} $r,s$ values.
The set of $k$-$(r,s)$ nuclei is found by the peeling algorithm~\cite{Sariyuce15} (given in \cref{alg:peeling}).
It is a generalization of the $k$-core and $k$-truss decomposition algorithms, and finds the $\kval_s$ indices of $r$-cliques in non-decreasing order.

The following lemma is standard in the $k$-core literature and we prove the analogue for $k$-$(r,s)$ nucleus.
It is a convenient characterization of the $\kval_s$ indices.

\begin{lemma} \label{lem:char1}
$\forall~R \in \cR(G)$, $\kval_s(R)=max_{\cR(G') \ni R}\delta_{r,s}(G')$, where $G' \subseteq G$.
\end{lemma}
\vspace{-3ex}
\begin{proof} Let $T$ be the $\kval_s(R)$-$(r,s)$ nucleus containing $R$.
By definition, $\delta_{r,s}(T) = \kval_s(R)$, so $\max_{G'} \delta_{r,s}(G') \geq \kval_s(R)$.
Assume the contrary that there exists some subgraph $T' \ni R$ such that $\delta_{r,s}(T') > \kval_s(R)$ (WLOG, we can assume $T'$ is connected; otherwise, we denote $T'$ to be the component containing $R$).
There must exist some maximal connected $T'' \supseteq T'$ that is a $\delta_{r,s}(T')$-nucleus.
This would imply that $\kval_s(R) \geq \delta_{r,s}(T') > \kval_s(R)$, a contradiction.
\end{proof}
\vspace{-2ex}

\subsection{{\large \boldmath{\lowercase{$h$}}}-index computation}
\noindent The main idea in our work is the iterative $h$-index computation on the $\cS$-degrees of $r$-cliques.
$h$-index metric is introduced to measure the impact and productivity of researchers by the citation counts~\cite{Hirsch05}.
A researcher has an $h$-index of $k$ if she has at least $k$ papers and each paper is cited at least $k$ times such that there is no $k'>k$ that satisfies these conditions.
We define the function $\cH$ to compute the $h$-index as follows:

\vspace{-1ex}
\begin{defn} \label{def:H}
Given a set $K$ of natural numbers, $\cH(K)$ is the largest $k \in \NN$ such that $\geq k$ elements of $K$ are $\geq k$.
\end{defn}
\vspace{-1ex}
Core number of a vertex can be defined as the largest $k$ such that it has at least $k$ neighbors whose core numbers are also at least $k$.  In the following section, we formalize this observation, and  build on it to design algorithms to compute not only core decompositions but also truss or nucleus decomposition for any $r$ and $s$ values.
\section{From the {\large {\lowercase{$\mathbf{h}$}}}-index to the {\lowercase {\Large {$\mathbf{\kval_s}$}}}-index} \label{sec:hindex}
\noindent Our main theoretical contribution is two-fold.
First, we introduce a generic formulation to compute the $k$-$(r,s)$ nucleus by an iterated $h$-index computation on $r$-cliques.
Secondly, we prove convergence bounds on the number of iterations.

We define the \emph{update operator} $\upf$.
This takes a function $\tau: \cR \to \NN$ and returns another function $\upf \tau: \cR \to \NN$, where $\cR$ is the set of $r$-cliques in the graph.

\begin{defn}
The update $\upf$ is applied on the $r$-cliques in a graph $G$ such that for each $r$-clique $R \in \cR(G)$:
\begin{compactenum}[\leftmargin=-5ex]
\item \label{step:one} \hspace{-1.5ex}1. For each $s$-clique $S\supset R$, set $\rho(S,R) = \min_{R' \subset S, R' \neq R} \tau(R')$.
\item \label{step:two} \hspace{-1.5ex}2. Set $\upf \tau(R) = \cH(\{\rho(S,R)\}_{S \supset R})$.
\end{compactenum}
\label{def:operator}
\end{defn}

Observe that $\upf \tau$ can be computed in parallel over all $r$-cliques in $\cR(G)$.
It is convenient to think of the $\cS$-degrees ($d_s$) and $\kval_s$ indices as functions $\cR \to \NN$.
We initialize $\tau_0 = d_s$, and set $\tau_{t+1} = \upf \tau_t$.

The results of Lu et al.~\cite{Lu16} prove that, for the $k$-core case ($r=1, s=2$), there is a sufficiently large $t$ such that $\tau_{t} = \kval_2$ (core number).
We generalize this result for \emph{any} nucleus decomposition.
Moreover, we prove the first convergence bounds for $\upf$.

The core idea of~\cite{Lu16} is to prove that the $\tau_t(\cdot)$ values
never increase (monotonicity) and are always lower bounded by core numbers.
We generalize their proof for any $(r,s)$ nucleus decomposition.

\vspace{-1ex}
\begin{theorem} \label{thm:prop} For all $t$ and all $r$-cliques $R$:
\begin{compactitem}
    \item (Monotonicity) $\tau_{t+1}(R) \leq \tau_t(R)$.
    \item (Lower bound) $\tau_t(R) \geq \kval_s(R)$.
\end{compactitem}
\end{theorem}
\vspace{-3ex}

\begin{proof} (Monotonicity) We prove by induction on $t$.
Consider the base case $t=0$. Note that for all $R$, $\tau_1(R) = \upf d_s(R) \leq d_s(R)$.
This is because in \Step{two}, the $\cH$ operator acts on a set of $d_s(R)$,
and this is largest possible value it can return.
Now for induction (assume the property is true up to $t$).
Fix an $r$-clique $R$, and $s$-clique $S \supset R$. 
For $\tau_t(R)$, one computes the value $\rho(S,R) = \min_{R' \subset S, R' \neq R} \tau_{t-1}(R)$.
By the induction hypothesis, the values $\rho(S,R)$ computed
for $\tau_{t+1}$ is at most the value computed for $\tau_t$.
Note that the $\cH$ operator is monotone; if one decreases values in a set $K$,
then $\cH(K)$ cannot increase. Since the $\rho$ values
cannot increase, $\tau_{t+1}(R) \leq \tau_t(R)$.

(Lower bound) We will prove that for any $G' \subseteq G$, $\cR(G') \ni R$, $\tau_t(R) \geq \delta_{r,s}(G')$.
\Lem{char1} completes the proof.

We prove the above by induction on $t$. For the base case,
$\tau_0(R) = d_s|_G(R) \geq d_s|_{G'}(R) \geq \delta_{r,s}(G')$. Now for induction.
By the induction hypothesis, $\forall~R \in \cR(G'), \tau_t(R) \geq \delta_{r,s}(G')$.
Consider the computation of $\tau_{t+1}(R)$, and the values
$\rho(S,R)$ computed in \Step{one}. For every $s$-clique $S$,
note that $\rho(S,R) = \min_{R' \subset S, R' \neq R}\tau_t(R')$.
By the induction hypothesis, this is at least $\delta_{r,s}(G')$.
By definition of $\delta_{r,s}(G')$, $d_s|_{G'}(R) \geq \delta_{r,s}(G')$.
Thus, in \Step{two}, $\cH$ returns at least $\delta_{r,s}(G')$.
\end{proof}

Note that this is an intermediate result and we will present our final result in \cref{lem:final} at the end.  

\subsection{Convergence bounds by the degree levels} \label{sec:level}
\noindent A trivial upper bound for convergence is the number of $r$-cliques in the graph, $|\cR(G)|$, because after $n$ iterations $n$ $r$-cliques with the lowest $\kval_s$ indices will converge.
We present a tighter bound for convergence.
Our main insight is to define the \emph{degree levels} of $r$-cliques, and relate these to the convergence of $\tau_t$ to $\kval_s$.
We prove that the $\kval_s$ indices in the $i$-th level converge within $i$ iterations of the update operation.
This gives quantitative bounds on the convergence.\\

\vspace{-3ex}
\begin{defn} \label{def:ring}
For a graph $G$,
\begin{compactitem}
\item $\cC(G)=\cR(G) \cup \cS(G)$, i.e., set of all $r$-cliques and $s$-cliques.
\item $S \in \cC(G)$ if and only if $R \in \cC(G), \forall~R \subset S$.
\item If $R$ is removed from $\cC(G)$, all $S \supset R$ are also removed from $\cC(G)$.
\item \textbf{\emph{Degree levels}} are defined recursively as follows.
The $i$-th level is the set $L_i$.
\begin{compactitem}[\leftmargin=0ex$-$]
    \item $L_0$ is the set of $r$-cliques that has the minimum $\sS$-degree in $\cC$.
    \item $L_i$ is the set of $r$-cliques that has the minimum $\sS$-degree in $\cC \setminus \bigcup_{j < i} L_j$.
\end{compactitem}
\end{compactitem}
\end{defn}

\cref{fig:toydeglevel} shows the degree levels for $k$-core decomposition on a toy graph. 
We first prove the $\kval_s$ indices cannot decrease as the level increases.
The following proof is closely related to the fact the minimum degree removal algorithm (peeling) finds all cores/trusses/nuclei.

\begin{theorem} \label{thm:mindeg}
Let $i \leq j$. For any $R_i \in L_i$ and $R_j \in L_j$, $\kval_s(R_i) \leq \kval_s(R_j)$.
\end{theorem}
\begin{proof} Let $L' = \bigcup_{r \geq i} L_r$, the union of all levels $i$ and above, and $G'$ is the graph such that $L'=\cR(G')$.
By definition of the levels, $d_s|_{G'}(R_i) = \delta_{r,s}(G')$ and $d_s|_{G'}(R_j) \geq d_s|_{G'}(R_i)$. 
There exists some $\kval_s(R_i)$-nucleus $T$ containing $R_i$.
We split into two cases.

\textbf{Case 1:} $\cR(T) \subseteq L'$.
Thus, $\kval_s(R_i) = \delta_{r,s}(T) \leq \delta_{r,s}(G')$ $= d_s|_{G'}(R_i)$.
Note that $\kval_s(R_j) = \min_{P \ni R_j} \delta_{r,s}(P)$, so $\kval_s(R_j) \geq \delta_{r,s}(G')$. Hence, $\kval_s(R_i) \leq \kval_s(R_j)$.

\textbf{Case 2:} $\cR(T) \setminus L' \neq \emptyset$.
Thus, there exists some $r$-clique $R' \in \cR(T) \cap L_b$, where $b < i$.
Choose the $R'$ that minimizes this value of $b$.
Since $T$ is a $\kval_s(R_i)$-nucleus, $d_s|_T(R') \geq \kval_s(R_i)$.
Consider $M = \bigcup_{r \geq b} L_r$. Note that $\cR(T) \subseteq M$,
since we chose $R'$ to minimize $b$.
Let $Q$ is the graph such that $M=\cR(Q)$.
We have $d_s|_Q(R') \geq d_s|_T(R') \geq \kval_s(R_i)$.
Since $R' \in L_b$, $d_s|_Q(R') = \delta_{r,s}(Q)$.
Since $j > b$ and $R_j \in M$, $\kval_s(R_j) \geq \delta_{r,s}(Q)$.
Combining the above, we deduce $\kval_s(R_i) \leq \kval_s(R_j)$.
\end{proof}

\begin{figure}[!t]
\centering
\vspace{1ex}
\includegraphics[width=0.7\linewidth]{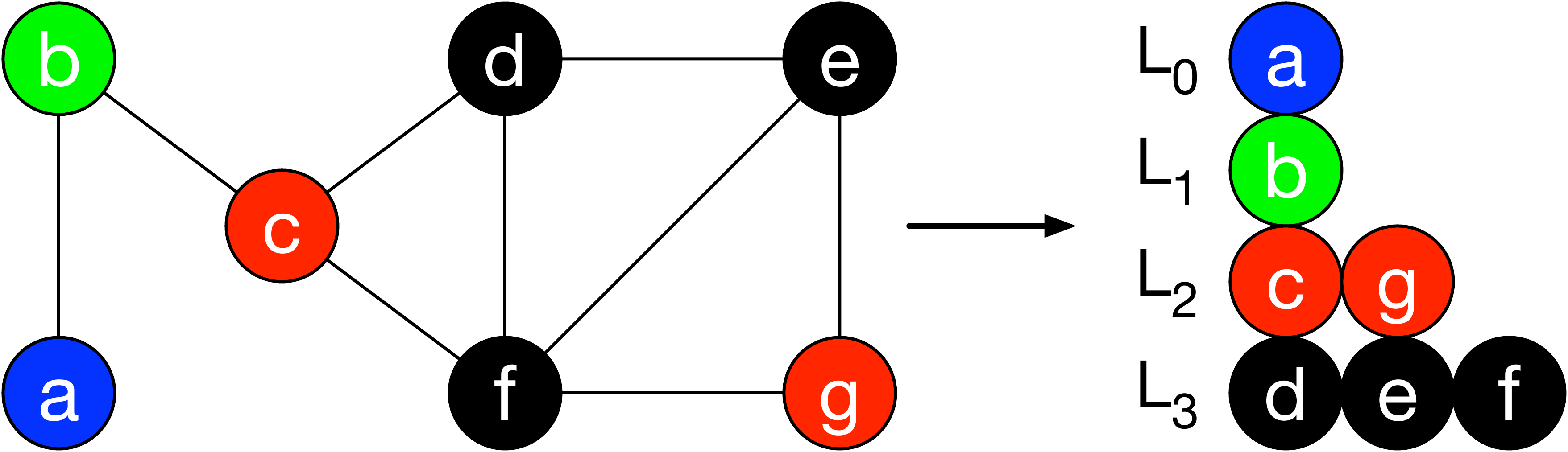}
\vspace{0ex}
\caption{\textit{Illustration of degree levels for the $k$-core decomposition. $L_0=\{a\}$ since it has the minimum degree initially and the only such vertex. Its removal makes the $b$ with the minimum degree, so $L_1=\{b\}$. After removing vertex $b$, there are two vertices with the least degree; $L_2=\{c,g\}$. Lastly, removing those leaves three vertices with the same degree and $L_3=\{d,e,f\}$.}}
\vspace{1ex}
\label{fig:toydeglevel}
\end{figure}

The main convergence theorem is the following. As explained earlier,
it shows that the $i$-th level converges within $i$ iterations.

\begin{theorem} \label{thm:conv} Fix level $L_i$.
For all $t \geq i$ and $R \in L_i$, $\tau_t(R) = \kval_s(R)$.
\end{theorem}
\begin{proof} We prove by induction on $i$.
For the base case $i\!=\!0$; note that for any $R$
of minimum $\cS$-degree in $G$, $\kval_s(R) \!=\! d_s|_G(R) \!= \!\tau_0(R)$.
For induction, assume the theorem is true up to level $i$.
Thus, for $t \geq i$ and $\forall~R \in \bigcup_{j \leq i} L_j$, $\tau_t(R) \!= \!\kval_s(R)$.
Select arbitrary $R_a\in L_{i+1}$,
and set $L' \!=\! \bigcup_{j \geq i+1} L_j$.

We partition the $s$-cliques containing $R_a$ into the ``low'' set $\cS_\ell$
and ``high'' set $\cS_h$. $s$-cliques in $\cS_\ell$ contain
some $r$-clique outside $L'$, and those in $\cS_h$ are contained
in $L'$. For every $s$-clique $S \in \cS_\ell$, there is a $R_b \subset S$
such that $R_b \in L_{k}$ for $k \!\leq \! i$. 
By the inductive hypothesis, $\tau_t(R_b) \!= \!\kval_s(R_b)$.
By \cref{thm:mindeg} applied to $R_b \in L_{k}$ and $R_a \!\in\! L_{i+1}$,
$\kval_s(R_b) \!\leq \!\kval_s(R_a)$. 

Now we focus on the computation of $\tau_{t+1}(R_a)$,
which starts with computing $\rho(S,R_a)$ in \Step{one} of \cref{def:operator}.
For every $S \in \cS_\ell$, by the previous argument,
there is some $r$-clique $R_b \subset S$, $R_b \neq R_a$, such
that $\tau_t(R_b) \! \leq \!\kval_s(R_a)$. Thus,
$\forall\,S\! \in\! \cS_\ell$, $\rho(S, R_a) \!\leq \!\kval_s(R_a)$.
This crucially uses the \emph{min} in the setting
of $\rho(S,R_a)$, and is a key insight into the generalization
of iterated $\cH$-indices for any nucleus decomposition.

The number of edges in $\cS_h$ is exactly 
$d_s|_{G'}(R_a) = \delta_{r,s}(G')$. Applying \Lem{char1}
to $R_a \in L'$, we deduce $\kval_s(R_a) \geq d_s|_{G'}(R_a)$.
All in all, for all $S \in \cS_\ell$,
$\rho(S,R_a)$ is at most $\kval_s(R_a)$. On the other hand,
there are at most $\kval_s(R_a)$ $s$-cliques in $\cS_h$.
The application of the $\cH$ function in \Step{two}
yields $\tau_{t+1}(R_a) \leq \kval_s(R_a)$.
But the lower bound of \cref{thm:prop} asserts $\tau_{t+1}(R_a) \geq \kval_s(R_a)$, and hence, these are equal. This completes the induction.
\end{proof}
\vspace{-1ex}

We have the following lemma to show that convergence is guaranteed in a finite number of iterations.
\begin{lemma}\label{lem:final}
Given a graph $G$ let $l$  be the maximum $i$, such that $L_l\neq \emptyset$  and  $\tau_l(R) \geq \kval_s(R)$ for all $r$-cliques (e.g., $\tau_0 = d_s$) and set $\tau_{t+1} = \upf \tau_t$.    For some $t\leq l$,
$\tau_t(R) = \kval_s(R)$, for all  $r$-cliques. 
\end{lemma}

\section{Local algorithms}

\noindent We introduce generalized local algorithms to find the $\kval_s$ indices of $r$-cliques for any $(r,s)$ nucleus decomposition.
For each $r$-clique, we iteratively apply $h$-index computation.
Our local algorithms are parallel thanks to the independent nature of the $h$-index computations.
We also explore time and quality trade-offs by using the iterative nature.
We first present the deterministic synchronous algorithm which does not depend on the order of processing the $r$-cliques.
It implements the $\upf$ operator in \cref{def:operator}.
Then we adapt our algorithm to work in an asynchronous manner that converges faster and uses less space.
For those familiar with linear algebra, the synchronous and asynchronous algorithms are analogous to Jacobi and Gauss-Seidel iterations for iterative solvers.
At the end, we discuss some heuristics and key implementation details for shared-memory parallelism in OpenMP.

\subsection{Synchronous Nucleus Decomposition (\textbf{\SND})}
\noindent We use the update operator $\upf$ to compute the $k$-$(r,s)$ nuclei of a graph $G$ in a synchronous way.
\cref{alg:SND} (\SND) implements the \cref{def:operator} for functions $\tau_0 = d_s$ and $\tau_{t+1} = \upf \tau_t$ to find the $\kval_s$ indices of $r$-cliques in graph $G$.

\SND algorithm iterates until no further updates occur for any $\tau$ index, which means all the $\tau$ indices converged to $\kval_s$.
Computation is performed synchronously on all the $r$-cliques and at each iteration $i$, $\tau_i$ indices are found for all $r$-cliques.
We declare two arrays, $\tau(\cdot)$ and $\tau^p(\cdot)$, to store the indices being computed and the indices that were computed in the previous iteration, respectively (Lines~\ref{sn:1} and~\ref{sn:4}).
$\tau(\cdot)$ are initialized to the $\cS$-degrees of the $r$-cliques since $\tau_0=d_s$ (\cref{sn:2}).
At each iteration, newly computed $\tau(\cdot)$ indices are backed up in $\tau^p(\cdot)$ (\cref{sn:7}), and the new $\tau(\cdot)$ indices are computed.
During the iterative process, convergence is checked by the flag $\mathcal{F}$ (\cref{sn:5}), which is initially set to \true~(\cref{sn:3}) and stays \true~as long as there is an update on a $\tau$ index (Lines~\ref{sn:6},~\ref{sn:13}, and~\ref{sn:14}).

\LinesNumbered
\begin{algorithm}[!t]
\caption{\small {\SND : Synchronous\,Nucleus\,Decomp}}\label{alg:SND}
\nonl \textbf{Input:} $G$: graph, $r, s$: positive integers ($r<s$)\;
\nonl \textbf{Output:} $\kappa_s(\cdot)$: array of $\kappa_s$ indices for $r$-cliques\;
$\tau(\cdot) \leftarrow$ indices $\forall~R \in \cR(G)$~\tcp*{current iteration}\label{sn:1}
$\tau(R) \leftarrow d_s(R)~\forall~R \in \cR(G)$~\tcp*{set to the $\cS$-degrees}\label{sn:2}
$\mathcal{F} \leftarrow \true$~\tcp*{stays \textnormal{\true} if any $\tau(R)$ is updated}\label{sn:3}
\algSND{${\tau^p(\cdot)} \leftarrow$ backup indices $\forall~R \in \cR(G)$~\tcp{prev. iter.}}\label{sn:4}
\While{$\mathcal{F}$} {\label{sn:5}
	$\mathcal{F} \leftarrow \false$\;\label{sn:6}
	\algSND{$\tau^p(R) \leftarrow \tau(R)~\forall~R \in \cR(G)$}\;\label{sn:7}
	\For{\meach $R \in \cR(G)$ in parallel} {\label{sn:8}
		$L \leftarrow$ empty set\;\label{sn:9}
		\For{\meach $s$\textnormal{-clique} $S \supset R$} {\label{sn:10}
			\algSND{${{\color{black}\rho \leftarrow \min_{R' \in \cN_s(R)}}}~\tau^p(R')$}\;\label{sn:11}
			$L$ . add ($\rho$)\;\label{sn:12}
		}
		\algSND{\If{$\tau^p(R) \neq \cH(L)$} {\label{sn:13}
			{{\color{black}$\mathcal{F} \leftarrow \true$}}\label{sn:14}
		}}
		$\tau(R) \leftarrow \cH(L)$\;\label{sn:15}
	}
}
$\kval_s(\cdot) \leftarrow \tau(\cdot)$\;\label{sn:16}
\Return array $\kval_s(\cdot)$\;\label{sn:17}
\end{algorithm}

\begin{figure}[!b]
\centering
\vspace{1ex}
\includegraphics[width=0.8\linewidth]{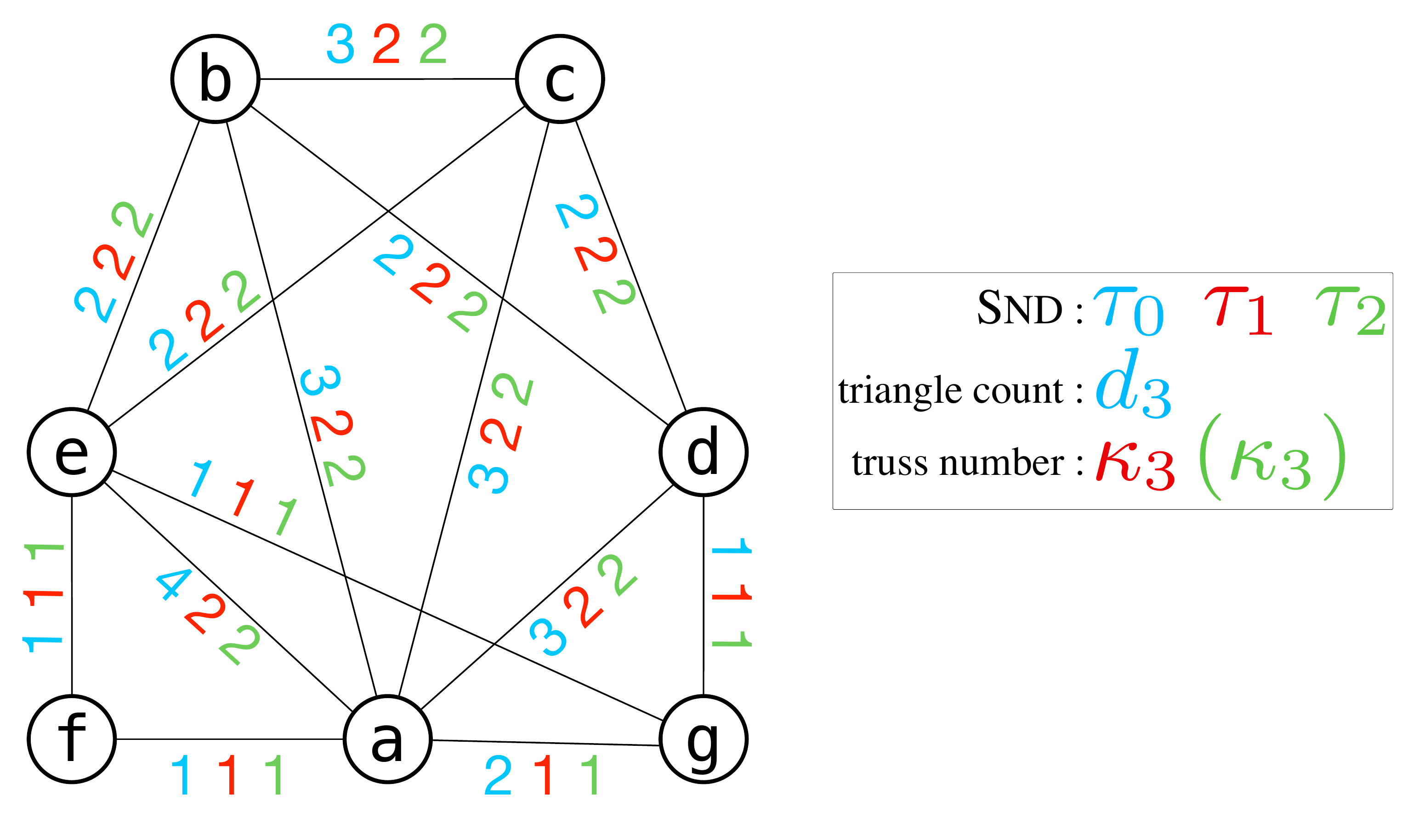}
\vspace{-1ex}
\caption{\it{\SND (\cref{alg:SND}) for the {$k$}-truss decomposition ($r=2,s=3$).
We find the $\kval_3$ indices.
Triangle counts of all the edges are computed ({$d_3$}) and set as their {$\tau_0$} values (blue).
For each edge, we first compute {$\tau_1$} indices (red) based on the {$\tau_0$} indices.
The \emph{bc} edge, for instance, has three triangles and for each of those we find the neighbor with the minimum {$\tau_0$} index and compute the {$h$}-index.
So, {$\tau_1(\emph{bc})=\cH\{min(\tau_0(\emph{ba}),\tau_0(\emph{ca})),min(\tau_0(\emph{bd}), \tau_0(\emph{cd})),min(\tau_0(\emph{be}),$\\$\tau_0(\emph{ce}))\}=\cH\{3, 2, 2\}=2$}.
No updates happen in the second iteration (green), so convergence is obtained in a single iteration.}}
\label{fig:trussEx}
\vspace{-5ex}
\end{figure}

Computation of the new $\tau(\cdot)$ indices for each $r$-clique can be performed in parallel (Lines~\ref{sn:8} to~\ref{sn:15}).
For each $r$-clique $R$, we apply the two step process in the \cref{def:operator}.
First, for each $s$-clique $S$ that contains $R$, we compute the $\rho$ values that is the minimum $\tau^p$ index of an $r$-clique $R' \subset S$ ($R'\neq R$) and collect them in a set $L$ (Lines~\ref{sn:10} to~\ref{sn:12}).
Then, we assign the $h$-index of the set $L$ as the new $\tau$ index of the $r$-clique (\cref{sn:15}). The algorithm continues until there are no updates on the $\tau$ index ` (Lines~\ref{sn:13} and~\ref{sn:14}).
Once the $\tau$ indices converge, we assign them to $\kappa_s$ indices and finish (Lines~\ref{sn:16} and~\ref{sn:17}).

\hl{
\textbf{Time complexity:} \SND algorithm starts with enumerating the $r$-cliques (not shown in the pseudocode) and its runtime is denoted by $RT_r(G)$ (this part can be parallelized as well, but we ignore that for now). Then, each iteration (Lines~\ref{sn:5} to~\ref{sn:15}) is performed $t$ times until convergence where $t$ is the total number of iterations for which we provided bounds in \cref{sec:level}. In each iteration, each $r$-clique $R \in \cR$ is processed once, \emph{which is parallelizable}. Suppose $R$ has vertices
$v_1, v_2, \cdots,v_r$ . We can find all $s$-cliques containing $R$ by looking at all $(s-r)$-tuples in
each of the neighborhoods of $v_i$ (Indeed, it suffices to look at just one such neighborhood). This takes $(\sum_R{\sum_{v\in R}{d(v)^{s-r}}})/p=(\sum_{v \in V}{\sum_{R \ni v}{d(v)^{s-r}}})/p=(\sum_{v \in V}{d_R(v)d(v)^{s-r}})/p$ time if $p$ threads are used for parallelism.
Note that the $h$-index computation can be done incrementally without storing all $\rho$ values in set $L$ (see~\cref{sec:heu}). Overall, the time complexity of \SND using $p$ threads is:

\begin{equation} \label{eq:comp}
{\small
\vspace{-1ex}
\cO\Big(RT_r(G)+t \big(\sum_{v \in V}{d_R(v)d(v)^{s-r}}\big) / p\Big)}
\vspace{-1ex}
\end{equation}

When $t=p$, complexity is same as the sequential peeling algorithm's  (\cref{alg:peeling}) and \SND is work-efficient.

\textbf{Space complexity:} In addition to the space that is needed to store $r$-cliques (taking $\cO(r\,|\cR(G)|)$), we need to store $\tau$ indices for the current and the previous iterations, which takes $\cO(|\cR(G)|)$ space, i.e., number of $r$-cliques. $\rho$ values are not need to be stored in set $L$ since the $h$-index computation can be done incrementally.
So, the total space complexity is $\cO(|\cR(G)|)$ (since $r=\cO(1)$).}

\begin{algorithm}[!t]
\caption{\small {\AND : Asynchronous\,Nucleus\,Decomp.}}
\label{alg:AND}
\nonl \textbf{Input:} $G$: graph, $r, s$: positive integers ($r<s$)\;
\nonl \textbf{Output:} $\kappa_s(\cdot)$: array of $\kappa_s$ indices for $r$-cliques\;
$\tau(\cdot) \leftarrow$ indices $\forall~R \in \cR(G)$~\tcp*{current iteration} \label{an:1}
$\tau(R) \leftarrow d_s(R)~\forall~R \in \cR(G)$~\tcp*{set to the $\cS$-degrees} \label{an:2}
$\mathcal{F} \leftarrow \textsc{true}$~\tcp*{stays \textnormal{\true} if any $\tau(R)$ is updated} \label{an:3}
\algSP{${\tt c}(R) \leftarrow \true~\forall~R \in \cR(G)$}\; \label{an:4}
\While{$\mathcal{F}$} { \label{an:5}
	$\mathcal{F} \leftarrow \false$\; \label{an:6}
	\For{\meach $R \in \cR(G)$ in parallel} { \label{an:7} 
		\algSP{\lIf{${\tt c}(R)$ is \textnormal{\false}}{cont.~{\bf else} ${\tt c}(R) \leftarrow \false$}} \label{an:8}
		$L \leftarrow$ empty set\; \label{an:9}
		\For{\meach $s$\textnormal{-clique} $S \supset R$} { \label{an:10}
			\algAND{${{\color{black}\rho \leftarrow \min_{R' \in \cN_s(R)}}}~\tau(R')$}\;\label{an:11}
			$L$ . add ($\rho$)\; \label{an:12}
		}
		\algAND{\If{$\tau(R) \neq \cH(L)$}{ \label{an:13}
			\algSP{${{\color{black}\mathcal{F} \leftarrow \true,}}~{\tt c}(R) \leftarrow \true$}\; \label{an:14}
			\algSP{\For{\meach $R' \in \cN_s(R)$} { \label{an:15}
			\If{$\cH(L) \le \tau(R')$}{ \label{an:16}
			${\tt c}(R')\leftarrow\true$} \label{an:17} }}
		}}
		$\tau(R) \leftarrow \cH(L)$\; \label{an:18}
	}
}
$\kval_s(\cdot) \leftarrow \tau(\cdot)$\;\label{an:19}
\Return array $\kval_s(\cdot)$\label{an:20}
\end{algorithm}

\cref{fig:trussEx} illustrates the \SND algorithm for $k$-truss decomposition ($r=2,s=3$) on a toy graph, where the participations of edges ($2$-cliques) in triangles ($3$-cliques) are examined.
Triangle counts of all the edges ($d_3$) are computed and set as their $\tau_0$ values (in blue).
For each edge, first we compute $\tau_1$ indices (in red) based on the $\tau_0$ indices (Lines~\ref{sn:5} to~\ref{sn:15}).
For instance, the {\bf ae} edge has four triangles and for each of those we find the neighbor with minimum $\tau_0$ index (Lines~\ref{sn:10} to~\ref{sn:12}); thus $L=\{(min(\tau_0({\bf eb}),\tau_0({\bf ab})),min(\tau_0({\bf ec}),$ $\tau_0({\bf ac})),min(\tau_0({\bf eg})$~$,\tau_0({\bf ag})),min(\tau_0({\bf ef}),\tau_0({\bf af}))\}=\{2, 2, 1,$ $1\}$ and $\tau_1({\bf ae})=\cH(L)=2$ (\cref{sn:15}).
Since the $\tau$ index is updated, we set flag $\cF$ \true~to continue iterations.
In the second iteration ($\tau_2$ indices), no update occurs, i.e., $\tau_2({\bf e})=\tau_1({\bf e})$ for all edges, thus the algorithm termnates.
So, one iteration is enough for the convergence and we have $\kval_3=\tau_1$ for all the edges.

\subsection{Asynchronous\,Nucleus\,Decomposition\,(\textbf{\AND})}\label{sec:async}

\noindent In the \SND algorithm, updates on the $\tau$ indices are synchronous and all the $r$-cliques are processed by using the same snapshot of $\tau$ indices.
However, when an $r$-clique $R$ is being processed in iteration $i$, a neighbor $r$-clique $R'\in\cN_s(R)$ might have already completed its computation in that iteration and updated its $\tau$ index.
By \cref{thm:prop}, we know that the $\tau$ index can only decrease as the algorithm proceeds.
Lower $\tau(R')$ indices in set $L$ can decrease $\cH(L)$, and it can help $\tau(R)$ to converge faster.
So, it is better to use the \textit{up-to-date} $\tau$ indices for faster convergence.
In addition, there would be no need to store the $\tau$ indices computed in the previous iteration, saving $|\cR(G)|$ space.

\begin{figure}[!t]
\centering
\vspace{-1ex}
\captionsetup[subfigure]{captionskip=-3ex}
{\includegraphics[width=0.9\linewidth]{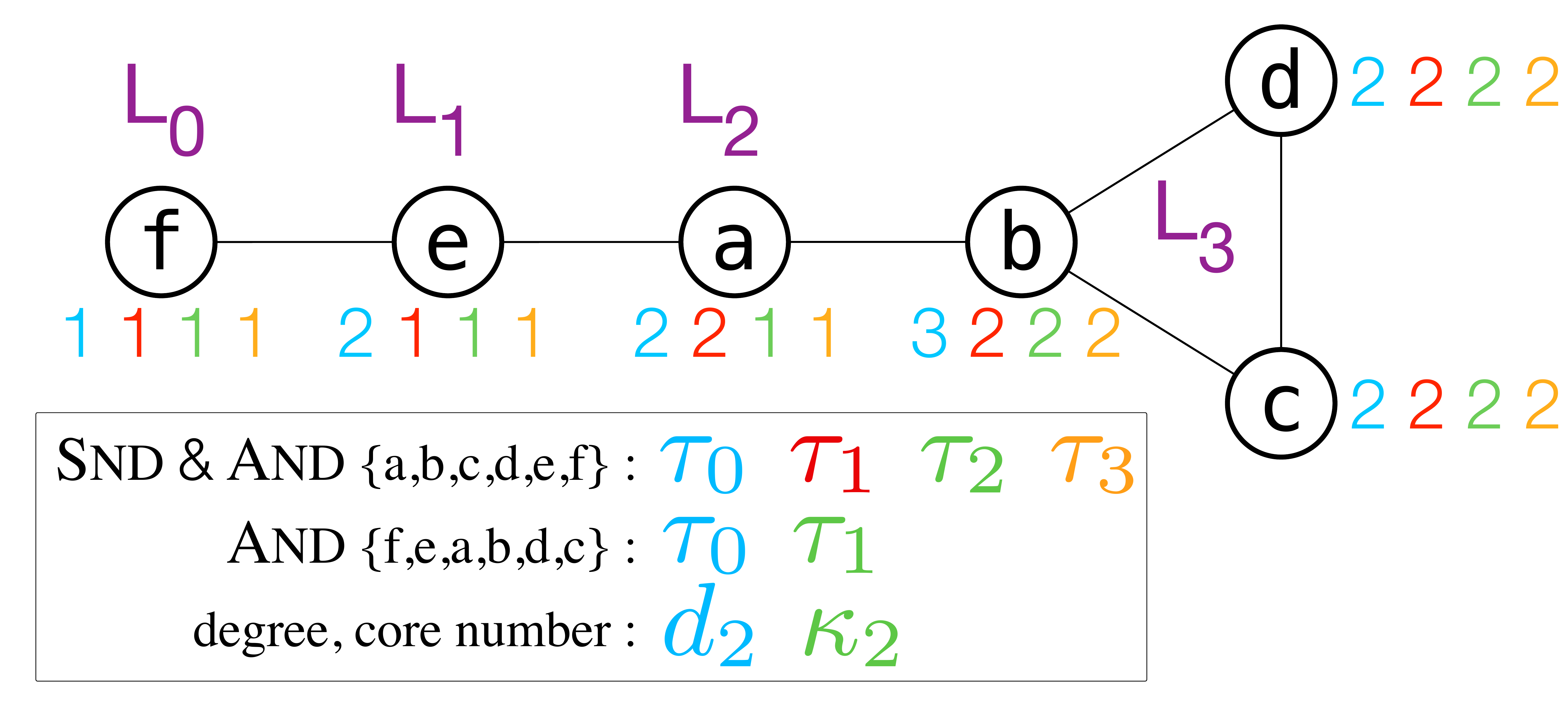}}
\vspace{-1ex}
\caption{\textit{\SND (\cref{alg:SND}) and \AND (\cref{alg:AND}, w/o orange lines) for the {$k$}-core decomposition ({$r=1,s=2$}). We find the {$\kval_2$} indices (core numbers) of vertices (edge is {$2$}-clique).
{$\tau_0$} indices are initialized to the degrees ({$d_2$} in blue). \SND algorithm uses the $\tau_{i-1}$ indices to compute the $\tau_i$ indices and converges in two iterations ({$\tau_1$} in red, {$\tau_2$} in green, {$\tau_3$} in orange). Same happens when we use \AND and follow the \{\emph{a,b,c,d,e,f}\} order to process the vertices.
On the other hand, if we choose the order by degree levels, \{\emph{f,e,a,b,c,d}\}, convergence is obtained in a single iteration.}}
\vspace{1ex}
\label{fig:core1}
\end{figure}

We introduce the \AND algorithm (\cref{alg:AND}) to leverage the up-to-date $\tau$ indices for faster convergence (ignore the orange lines for now).
At each iteration, we propose to use the latest available information in the neighborhood of an $r$-clique.
Removing the green lines in the \SND algorithm and inserting the blue lines in the \AND algorithm are sufficient to switch from synchronous to asynchronous computation.
We do not need to use the $\tau_p(\cdot)$ to back up the indices in the previous iteration anymore, so Lines~\ref{sn:4} and~\ref{sn:7} in \cref{alg:SND} are removed.
Computation is done on the latest $\tau$ indices, so we adjust the Lines~\ref{sn:11} and~\ref{sn:13} in \cref{alg:SND} accordingly, to use the up-to-date $\tau$ indices.

In the same iteration, each $r$-clique can have a different view of the $\tau(\cdot)$ and updates are done \textit{asynchronously} in an arbitrary order.
Number of iterations for convergence depends on the processing order (\cref{an:7} in \cref{alg:AND}) and never more than the \SND algorithm.

\hl{
\textbf{Time complexity:} \textit{The worst case for \AND happens when all the $r$-cliques see the $\tau$ indices of their neighbors that are computed in the previous iteration, which exactly corresponds to the \SND algorithm.} Thus the time complexity of \AND is same as \SND's (\cref{eq:comp}). However, in practice we expect fewer iterations.

\textbf{Space complexity:} The only difference with \SND is that we do not need to store the $\tau$ values in the previous iteration anymore. So, it is still $\cO(|\cR(G)|)$.
}

\cref{fig:core1} illustrates \AND algorithm with two different orderings and the \SND algorithm on the $k$-core case ($r=1,s=2$).
Focus is on vertices ($1$-cliques) and their edge ($2$-clique) counts (degrees).
We start with \SND.
Vertex degrees are set as $\tau_0$ indices (blue).
For each vertex $u$ we compute the $\tau_1(u)=\cH(\{\tau_0(v) : v \in \cN_2(u) \}$ (red), i.e., $h$-index of its neighbors' degrees.
For instance, vertex \textbf{a} has two neighbors, \textbf{e} and \textbf{b}, with degrees $2$ and $3$.
Since $\cH(\{2,3\})=2$, we get $\tau_1($\textbf{a}$)=2$.
For vertex \textbf{b}, we get $\tau_1($\textbf{b}$)=\cH(\{2,2,2\})=2$.
After computing all the $\tau_1$ indices, $\tau$ values of vertices {\bf e} and {\bf b} are updated, thus we compute the $\tau_2$ indices, shown in green.
We observe an update for the vertex \textbf{a}; $\tau_2($\textbf{a}$)=\cH(\{\tau_1($\textbf{e}$), \tau_1($\textbf{b})$\})=\cH(\{1,2\})=1$ and continue computation.
For $\tau_3$ indices (orange), no update is observed which means that $\kval_s=\tau_2$,  and \SND converges in two iterations.
Regarding the \AND algorithm, say we choose to follow the increasing order of degree levels (noted in purple) where $L_0=\{\textbf{f}\}$, $L_1=\{\textbf{e}\}$, $L_2=\{\textbf{a}\}$, $L_3=\{\textbf{b,c,d}\}$.
Computing the $\tau_1$ indices on this order enables us to reach the convergence in a single iteration.
For instance, $\tau_1($\textbf{a}$)=\cH(\{\tau_1($\textbf{e}$), \tau_0($\textbf{b})$\})=\cH(\{1,3\})=1$.
However, if we choose to process the vertices in a different order than the degree levels, say \{\textbf{a,b,c,d,e,f}\}, we have $\tau_1($\textbf{a}$)=\cH(\{\tau_0($\textbf{e}$), \tau_0($\textbf{b})$\})=\cH(\{2,3\})=2$, and need more iteration(s) to converge.
Indeed, \textbf{a} is the only updated vertex.
In the second iteration, we get $\tau_2($\textbf{a}$)=\cH(\{\tau_1($\textbf{e}$), \tau_1($\textbf{b})$\})=\cH(\{1,2\})=1$, an update, thus continue iterating.
Third iteration does not change the $\tau$ indices, so \AND with \{\textbf{a,b,c,d,e,f}\} order converges in two iterations, just as the \SND.\\

\begin{figure}[!t]
\centering
\vspace{-15ex}
\includegraphics[width=\linewidth]{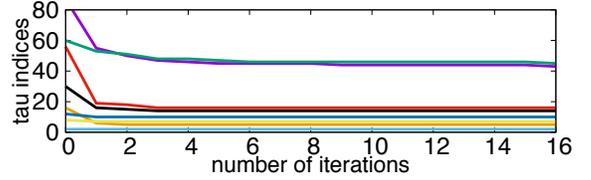}
\vspace{-18ex}
\caption{\textit{Changes in {$\tau$} indices of randomly selected edges in \fb graph during the {$k$}-truss decomposition. Wide plateaus appear during the convergence, especially at the end.
}}
\label{fig:plato}
\vspace{1ex}
\end{figure}

\vspace{-2ex}
\subsection{Avoiding redundant work by notifications}\label{sec:notif}
\noindent \SND and \AND algorithms converge when no $r$-clique updates its $\tau$ index anymore.
Consequently, update on all  $r$-cliques continue even when only one update occurs and we need an extra iteration to detect convergence.
\cref{fig:plato} shows the $\tau$ indices of randomly selected edges in the \fb graph during $k$-truss decomposition ($r=2,s=3$).
There are plenty of wide plateaus where $\tau$ indices stay constant, which implies  redundant computations.
How can we avoid this redundancy? 
Observe that repeating $\tau$ indices or plateaus are not sufficient, because an update can still occur after maintaining the same $\tau$ index for a number of iterations, creating a plateau.
In order to efficiently detect the convergence and skip any plateaus during the computation, we introduce a notification mechanism where an $r$-clique is notified to recompute its $\tau$ index, if any of its neighbors has an update.

Orange lines in \cref{alg:AND} present the notification mechanism added to  \AND . 
${\tt c}(\cdot)$ array is declared in (\cref{an:4}) to track whether an $R \in \cR(G)$ has updated its $\tau$ index.
${\tt c}(R)=\false$ means $R$ did not update its $\tau$ index, it is an idle $r$-clique, and there is no need to recompute its $\tau$ value for that iteration (\cref{an:8}). A non-idle $r$-clique is called active.
Thus, all ${\tt c}(\cdot)$ is set to \true~at the beginning to initiate the computations for all $r$-cliques.
Each $r$-clique marks itself idle at the beginning of an iteration (\cref{an:8}) and waits for an update in a neighbor.
When the $\tau(R)$ is updated, $\tau$ indices of \emph{some} neighbor $r$-cliques in $\cN_s(R)$ might be affected and they should be notified.
If $R' \in \cN_s(R)$, if $\tau(R') < \cH(L)$ (new $\tau(R)$) then $\tau(R') \le \tau(R)$ already in the previous iteration (\cref{thm:prop}), and thus no change can happen in the $h$-index computation.
Therefore, we only need to notify the neighbors that have $\tau$ indices greater than or equal to $\cH(L)$ (Lines~\ref{an:15} to~\ref{an:17}).
This version of our algorithm  requires  an additional $O( |\cR(G)|)$ space for ${\tt c}(\cdot)$ array and does not offer a theoretical improvement in time-complexity. However, avoiding redundant computations  yields faster runtimes in practice. 

\cref{fig:core2} illustrates the notification mechanism on the graph in \cref{fig:core1}, processing  the vertices in the  {\bf  a,b,c,d,e,f} order.
Again, vertex degrees are set as $\tau_0$ indices (blue) and we compute $\tau_1(u)=\cH(\{\tau_0(v) : v \in \cN_2(u) \}$, i.e., $h$-index of its neighbors' degrees, (red) for each vertex $u$.
No update happens for vertex {\bf a} and no vertices are notified.
$\tau({\bf b})$ is updated as $2$ and we check if any neighbors of {\bf b} has a $\tau$ index $\ge 2$ (\cref{an:16}).
All its neighbors have such $\tau$ indices, thus all are notified: {\bf a, c, d}.
Vertices {\bf c} and {\bf d} do not update their $\tau$ indices.
Then, $\tau({\bf e})$ is updated as $1$ and since $\tau_0({\bf e})\ge\tau_1({\bf a})>\tau_1({\bf e})$, vertices {\bf a} and {\bf f} are notified for recomputing its $\tau$ index.
At that instant, vertices {\bf a} and {\bf f} are active.
Next, vertex {\bf f} is processed and does not change its $\tau$ index, so all the vertices except {\bf a} are idle now.
In the second iteration, we only process {\bf a} and compute $\tau_2({\bf a})=\cH\{\tau_1({\bf e}), \tau_1({\bf b})\}=\cH\{1,2\}=1$.
Update in $\boldmath{\tau({\bf a})}$ notifies vertices {\bf b} and {\bf e} since both have $ \geq \tau$ indices.
In the third iteration, we recompute $\tau$ indices for {\bf b} and {\bf e}, but there is no update. So all the vertices become idle, implying  convergence.
Overall, it takes 9 $\tau$ computations  and 3 iterations   for the \AND with notification mechanism, while  24 $\tau$ computations   and 4 iterations are needed  without  the notification mechanism  (\cref{fig:core1}).
So the notification mechanism is helpful to avoid redundant computations.\\

\begin{figure}[!t]
\centering
\vspace{-1ex}
\captionsetup[subfigure]{captionskip=-3ex}
{\includegraphics[width=0.9\linewidth]{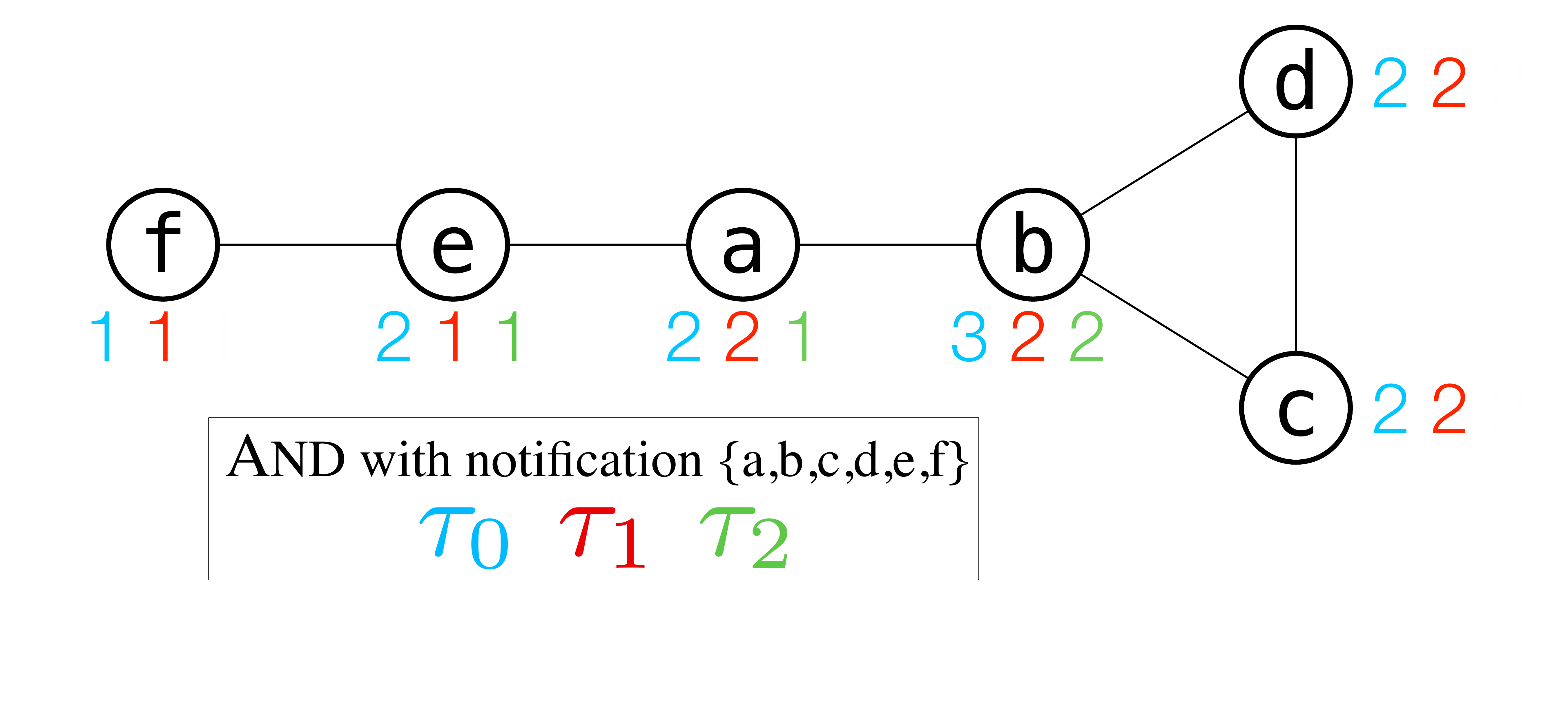}}
\vspace{-5ex}
\caption{\textit{{$k$}-core decomposition ({$r=1,s=2$}) by \AND (\cref{alg:AND}) that uses the notification mechanism.
After the first iteration, the only active vertex is $a$. In the second iteration, computation updates $\tau(a)$ and thus notifies vertices $b$ and $e$. In the third iteration, their $\tau$ indices are recomputed and no update happens. All the vertices become idle, thus convergence is obtained.
9 $\tau$ computations performed in 3 iterations by \AND with notification mechanism, while 24 $\tau$ computations are done in 4 iterations if notification mechanism is not used (\cref{fig:core1}).}}
\label{fig:core2}
\vspace{1ex}
\end{figure}

\vspace{-1ex}
\hl{
\noindent \textbf{\PAND on a set of $r$-cliques:} Local nature of the \AND algorithm enables selection of  a set of $r$-cliques and its application  only to  this set until convergence.
This is useful in query-driven scenarios where the focus is on a single (or a few) vertex/edge.
We define \PAND as the application of \AND algorithm on a set of given $r$-cliques, say $P$.
 We only  modify  the orange lines in \cref{alg:AND} where {\tt c} of an $r$-clique is re-computed only if it is in set $P$.
This way we just limit the \AND computation on a small set.
We give an application of \PAND in~\cref{sec:app} where the computation is limited on a given $r$-clique and its neighbors.
}

\subsection{Heuristics and implementation}\label{sec:heu}
\noindent We introduce key implementation details for the shared-memory parallelism and heuristics for efficient $h$-index computation.
We used OpenMP~\cite{openmp} to utilize the shared-memory architectures.
The loops, annotated as parallel in Algorithms \ref{alg:SND} and \ref{alg:AND}, are shared among the threads, and each thread
is responsible for its partition of $r$-cliques.
No synchronization or atomic operation is needed.
Default scheduling policy in OpenMP is \textit{static} and it distributes the iterations of the loop to the threads in chunks, i.e., for two threads, one takes the first half and the other takes the second.
This approach  does not work well for our algorithms, since the notification mechanism may result in significant load imbalance among threads.
If most of the idle $r$-cliques are assigned to a certain thread, this thread quickly finishes, and remains idle until the iteration ends.
To prevent this, we adopted the dynamic scheduling where each thread is given a new workload once it idle.
We set chunk size to 100 and observed no significant difference for other values.
No thread stays idle this way,  improving  parallel efficiency.

$h$-index computation for a list of numbers is traditionally done by sorting the numbers in the non-increasing order and checking the values starting from the head of the list to find the largest $h$ value for which at least $h$ items exist with at least $h$ value.
Main bottleneck in this routine is the sorting operation which takes $\cO(nlog n)$ time for $n$ numbers.
We use a linear time algorithm that uses a hashmap and does not include sorting to compute the $h$-index.
$h$ is initialized as zero and we iterate over the items in the list.
At each step, we attempt to increase the current $h$ value based on the inspected item.
For the present $h$ value in a step, we keep track of the number of items examined so far that have value equal to $h$.
We use a hashmap to keep track of the number of items that has at least $h$ value, and ignore values smaller than $h$.
This enables the computation of the $h$-index in linear time and provides a trade-off between time and space.
In addition, after the initialization, we  check to see if the current $\tau$ index can be preserved.
Once we see at least $\tau$ items with index at least $\tau$, no more checks needed.

\begin{table}[!t]
\centering
\small 
\renewcommand{\tabcolsep}{3.5pt}
%
\caption{{\it Statistics about our dataset; number of vertices, edges, triangles and four-cliques ({$K_4$}).}}
\vspace{-2ex}
\begin{tabular}{|r|r|r|r|r|r|r|r|r|r|r|r|}\hline
	&	$|V|$			&	$|E|$			&	$|\triangle|$		 	&	$|K_4|$			\\ \hline
\assk 	(\tassk)	&$	1.7	$M	&$	11.1	$M	&$	28.8	$M	&$	148.8	$M	\\ \hline
\fb (\tfb)	&$	4	$K		&$	88.2	$K	&$	1.6	$M	&$	30.0	$M	\\ \hline
\hl{\fri (\tfri)} & $65.6$M & $1.8$B & $4.1$B & $8.9$B \\\hline
\slj 	(\tslj)	&$	4.8	$M	&$	68.5	$M	&$	285.7	$M	&$9.9	 	$B		\\ \hline
\so	(\tso)	&$	2.9	$M	&$	106.3	$M	&$	524.6	$M	&$	2.4	$B	\\ \hline
\sse	(\tsse)	&$	131.8	$K	&$	711.2	$K	&$	4.9	$M	&$	58.6	$M	\\ \hline
\sth	(\tsth)	&$	456.6	$K	&$	12.5$M	&$	83.0	$M	&$	429.7	$M	\\ \hline
\tw (\ttw)	&$	81.3	$K	&$	1.3	$M	&$	13.1	$M	&$	104.9	$M	\\ \hline
\wg (\twg)	&$	916.4	$K	&$	4.3	$M	&$	13.4$M	&$	39.9	$M	\\ \hline
\wn (\twn)	&$	325.7	$K	&$	1.1	$M	&$	8.9$M	&$	231.9$M	\\ \hline
\wiki (\twiki)	&$	3.1	$M	&$	37.0	$M	&$	88.8$M	&$	162.9$M	\\ \hline
\end{tabular}
\vspace{2ex}
\label{tab:dataset}
\end{table}

\section{Experiments}
\noindent We evaluate our algorithms on three instances of the $(r,s)$ nucleus decomposition: $k$-core ($(1,2)$), $k$-truss ($(2,3)$), and $(3,4)$ nucleus, which are shown to be the practical and effective~\cite{Sariyuce15, Sariyuce-TWEB17}.
We do not store the $s$-cliques during the computation for better scalability in terms of the memory space.
Instead, we find the participations of the $r$-cliques in the $s$-cliques on-the-fly~\cite{Sariyuce15}.
Our dataset contains a diverse set of real-world networks from SNAP~\cite{snap} and Network Repository~\cite{nr} (see \cref{tab:dataset}), such as internet topology network (\assk), social networks (\fb, \fri, \slj, \so), trust network ({\tt soc-} {\tt sign-epinions}), Twitter follower-followee networks ({\tt soc-twitter} {\tt -higgs}, \tw), web networks (\wg, {\tt web-} {\tt NotreDame}), and a network of wiki pages (\wiki).

All experiments are performed on a Linux operating system running on a machine with Intel Ivy Bridge processor at 2.4 GHz with 64 GB DDR3 1866 MHz memory.
There are two sockets on the machine and each has 12 cores, making 24 cores in total.
Algorithms are implemented in C++ and compiled using gcc 6.1.0 at the -O2 level.
We used OpenMP v4.5 for the shared-memory parallelization.
Code is available at \url{http://sariyuce.com/pnd.tar}.

We first investigate the convergence characteristics of our new algorithms in \cref{sec:conAna}.
We compare the number of iterations that our algorithms need for the convergence and also examine the convergence rates for the $\kappa$ values. In addition, we investigate how the densest subgraphs evolve and present a metric that can be monitored to determine the ``good-enough'' decompositions so that trade-offs between quality and runtime can be enjoyed.
Then, we evaluate the runtime performance in \cref{sec:runtime}. In particular, we examine the impact of notification mechanism (\cref{sec:notif}) on the \AND algorithm, show the scalability for our best performing method, and compare it with respect to the partially parallel peeling algorithms.
We also examine the runtime and accuracy trade-off for our algorithms.
Last, but not least, we highlight a query-driven scenario in \cref{sec:app} where our algorithms are used on a subset of vertices/edges to estimate the core and truss numbers.

%

\subsection{Convergence Analysis}\label{sec:conAna}
\noindent Here we study the following questions:\\

\begin{compactitem}[\leftmargin=0.1ex $\bullet$]
\item How does the number of iterations change between  asynchronous computation (\AND)  and synchronous (\SND)? How do they relate to our theoretical bounds of \cref{sec:level}?
\item What is the rate of convergence regarding the $\tau$ values? How quickly do they approach to the $\kappa$ values?
\item How is the evolution of the important subgraphs (with high density) during the convergence?
\item Is there a generic way to infer the ``good-enough'' decompositions so that the computation can be halted for  trade-off between runtime and quality?
\end{compactitem}

\begin{table}[!t]
\vspace{0ex}
\centering
\renewcommand{\tabcolsep}{3.5pt}
\caption{\it Number of iterations for the theoretical upper bound, Degree Levels \textnormal{(DL)}(\cref{sec:level}), \SND, and \AND algorithms. 
}
\label{tab:converge}
\vspace{-1ex}
\begin{tabular}{|c|r|r|r|r|r|r|r|r|}\hline
		&	&	\tassk			&	\tfb			&	\tslj			&	\tso			&	\tsth			&	\twg			&	\twiki			\\ \hline
\multirow{3}{*}{$k$-core}	&	DL	&	$	1195	$	&	$	352	$	&	$	3479	$	&	$	5165	$	&	$	1713	$	&	$	384	$	&	$	2026	$	\\ \cline {2-9}
	&	\SND	&	$	63	$	&	$	21	$	&	$	99	$	&	$	147	$	&	$	73	$	&	$	23	$	&	$	55	$	\\ \cline {2-9}
	&	\AND	&	$	33	$	&	$	11	$	&	$	51	$	&	$	73	$	&	$	37	$	&	$	14	$	&	$	30	$	\\ \hline 
\multirow{3}{*}{$k$-truss}	&	DL	&	$	1605	$	&	$	859	$	&	$	5401	$	&	$	4031	$	&	$	2215	$	&	$	254	$	&	$	2824	$	\\ \cline {2-9}
	&	\SND	&	$	118	$	&	$	33	$	&	$	86	$	&	$	207	$	&	$	101	$	&	$	20	$	&	$	562	$	\\ \cline {2-9}
	&	\AND	&	$	58	$	&	$	19	$	&	$	44	$	&	$	103	$	&	$	53	$	&	$	11	$	&	$	410	$	\\ \hline 
\multirow{3}{*}{$(3,4)$}	&	DL	&	$	1734	$	&	$	1171	$	&	$	7426	$	&	$	3757	$	&	$	2360	$	&	$	157	$	&	$	1559	$	\\ \cline {2-9}
	&	\SND	&	$	72	$	&	$	38	$	&	$	123	$	&	$	196	$	&	$	109	$	&	$	11	$	&	$	122	$	\\ \cline {2-9}
	&	\AND	&	$	41	$	&	$	23	$	&	$	73	$	&	$	116	$	&	$	51	$	&	$	6	$	&	$	107	$	\\ \hline
\end{tabular}
\vspace{1ex}
\end{table}

\subsubsection{Number of iterations}
As described in \cref{sec:async}, the number of iterations for convergence can (only) be decreased by the asynchronous algorithm \AND.
We compare \SND (\cref{alg:SND}) and \AND (\cref{alg:AND}) for three nucleus decompositions.
All runs are performed in sequential, and for \AND we use the natural ordering of the $r$-cliques in datasets that is the order of vertices/edges/triangles given or computed based on the ids in the data files.
Note that, we also checked \AND with random $r$-clique orderings and did not observe significant differences.
We also compute the number of degree levels (\cref{def:ring}) that we prove as an upper bound in~\cref{sec:level}.

\cref{tab:converge} presents the results for $k$-core, $k$-truss, and $(3,4)$ nucleus decompositions.
Number of degree levels gives much tighter bounds than the obvious limits -- number of $r$-cliques.
We observe that both algorithms converge in far fewer iterations than our upper bounds -- \SND converges in $5$\% of the bounds given for all decompositions, on average.
Regarding the comparison, \AND algorithm converges in $50$\% fewer iterations than the \SND for $k$-core and $k$-truss decompositions and $35$\% fewer iterations for $(3,4)$ nucleus decomposition.
Overall, we see the clear benefit of asynchronous computation on all the decompositions, thus use \AND algorithm in the following experiments.

\begin{figure}[!t]
\centering
\captionsetup[subfigure]{captionskip=-10ex}
\vspace{-14ex}
\subfloat[{$k$}{-core}]{\includegraphics[width=0.95\linewidth]{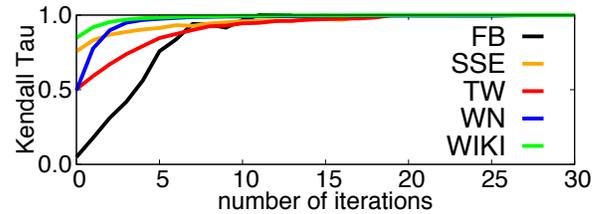}}
\vspace{-23ex}
\subfloat[{$(3,4)$} {nucleus}]{\includegraphics[width=0.95\linewidth]{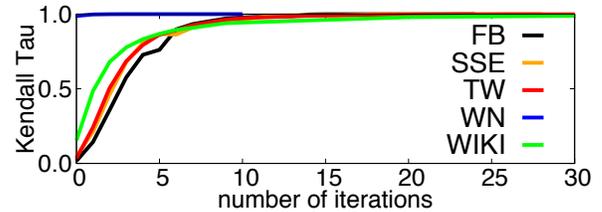}}
\vspace{-9ex}
\caption{\textit{Convergence rates for five graphs in our dataset. Kendall-Tau similarity score compares the {$\tau$} values in a given iteration with the exact decomposition ({$\kappa$} values); becomes 1.0 when they  are equal. Our algorithms compute almost-exact decompositions in around 10 iterations for {$k$}-core, {$k$}-truss (in \cref{fig:intro}), and {$(3,4)$} nucleus decompositions.
}}
\vspace{-1ex}
\label{fig:KT}
\end{figure}

\subsubsection{Convergence rates for the $\tau$ values}
In the previous section, we studied the number of iterations  required for exact solutions.  Now we will investigate how fast our estimates, $\tau$ values converge to the exact, $\kappa$ values.
We use Kendall-Tau similarity score to compare the $\tau$ and $\kappa$ values for each iteration, which becomes 1, when they are equal.  
\cref{fig:KT} and~\cref{fig:intro} present the results for five representative graphs in our dataset.
We observe that our local algorithms compute almost exact decompositions in less than 10 iterations for all decompositions, and we need $5$, $9$, and $6$ iterations to achieve $0.90$ similarity for $k$-core, $k$-truss, and $(3,4)$ nucleus decompositions, respectively.

\begin{figure}[!t]
\centering
\vspace{-9ex}
\captionsetup[subfigure]{justification=centering, captionskip=-4ex}
\subfloat[$k$-core on \tfb]{\hspace{-3ex}\includegraphics[width=0.56\linewidth]{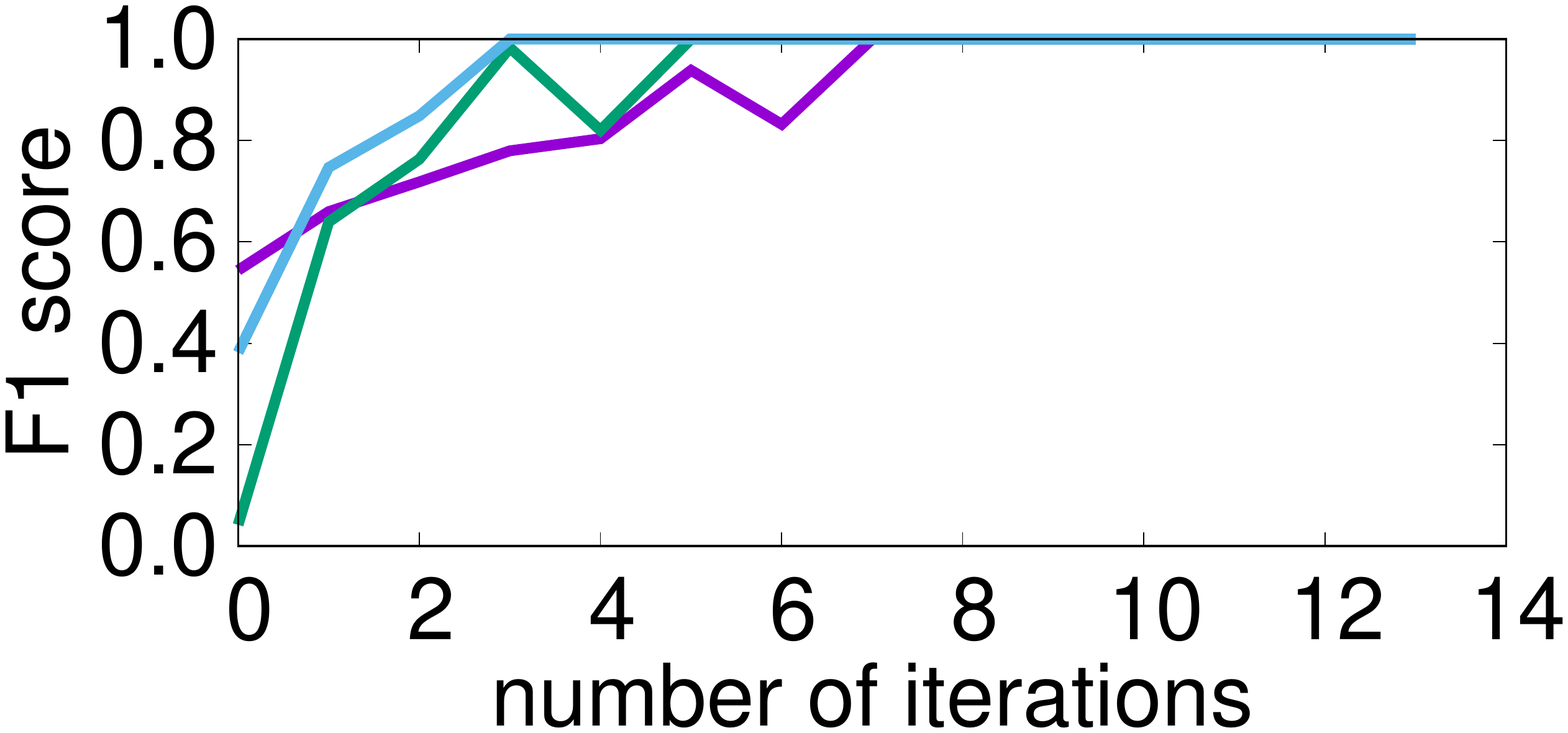}\hspace{-3ex}}
\subfloat[$k$-core on \twn]{\includegraphics[width=0.56\linewidth]{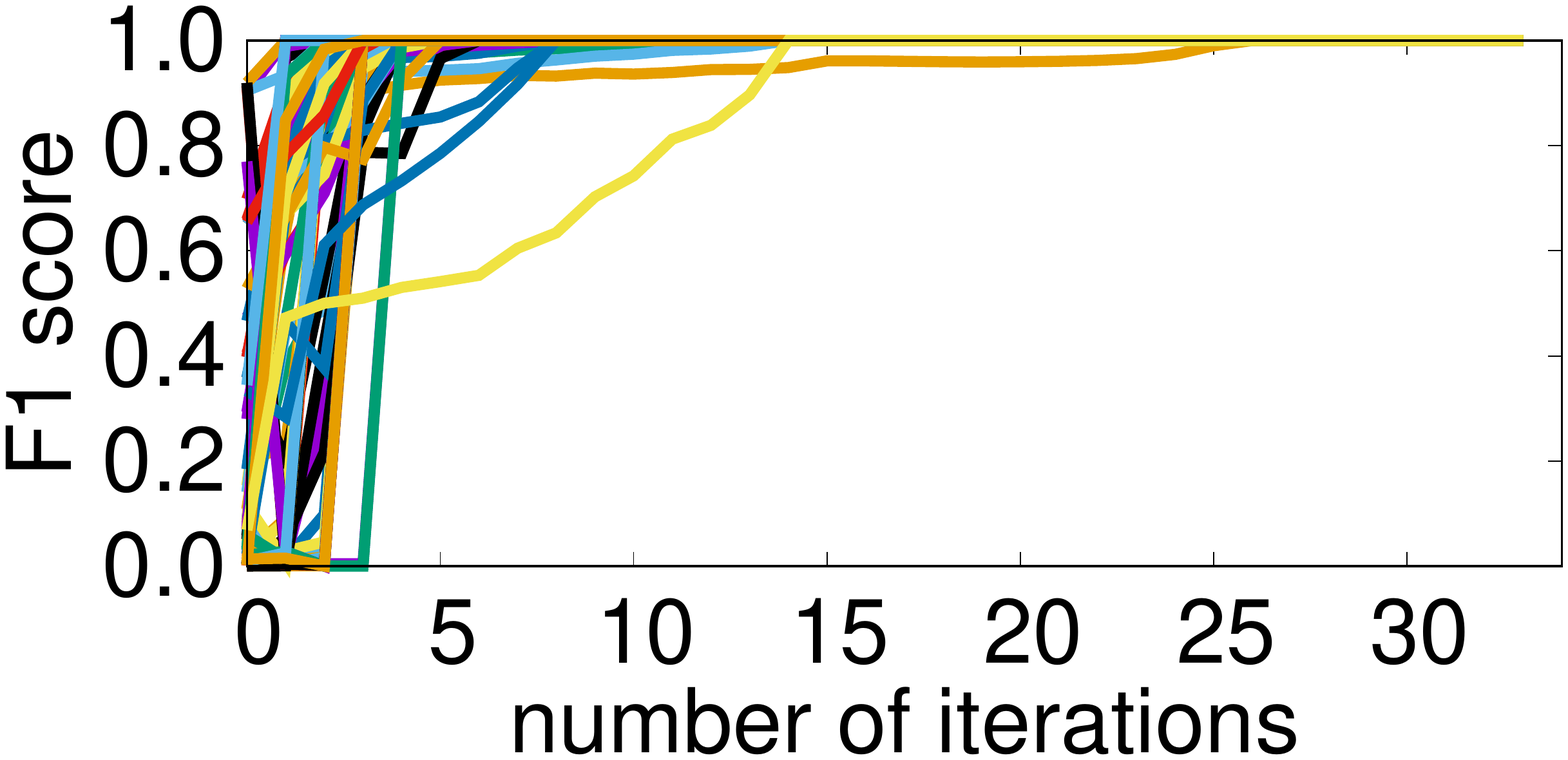}}
\vspace{-10ex}

\subfloat[$k$-truss on \tfb]{\hspace{-3ex}\includegraphics[width=0.56\linewidth]{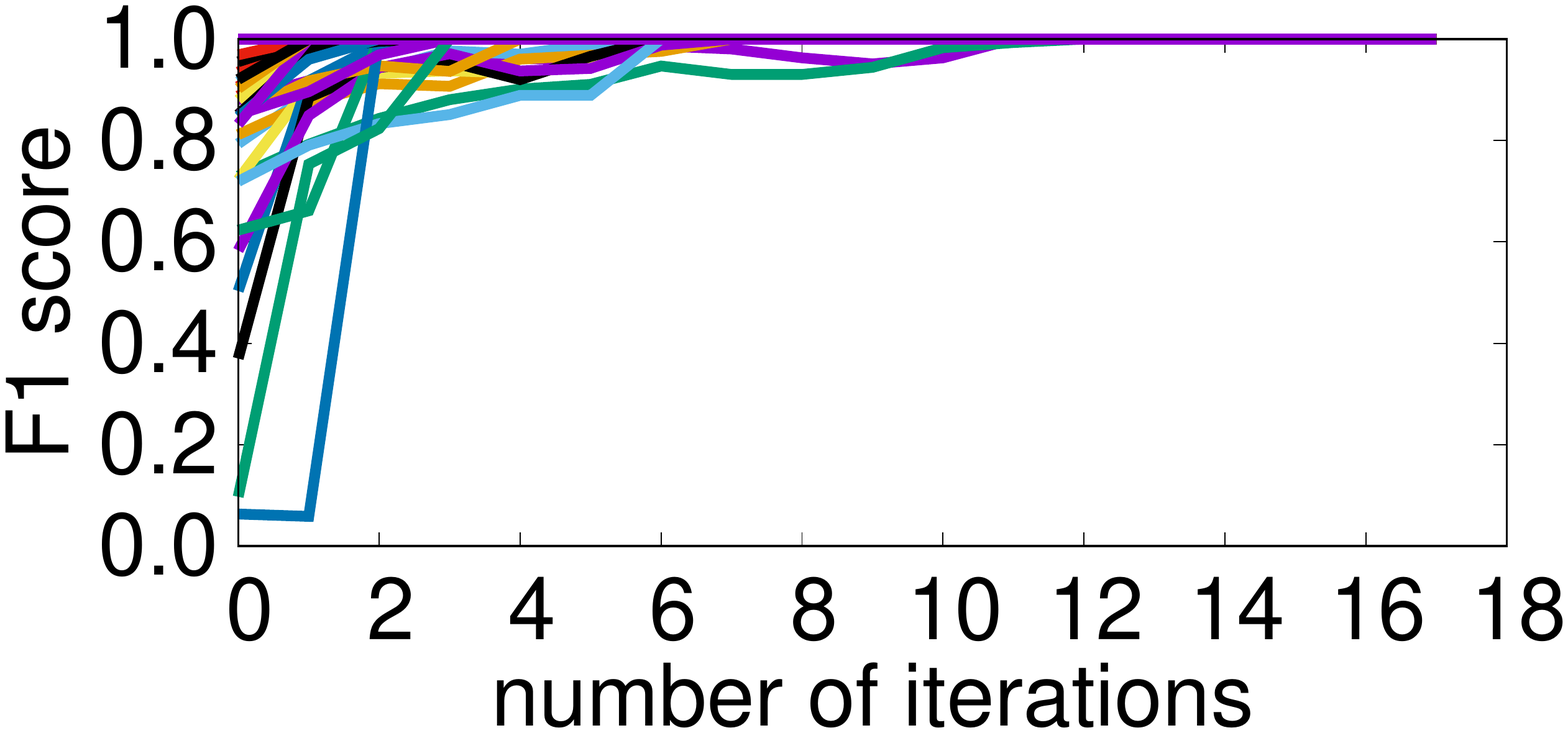}\hspace{-3ex}}
\subfloat[$k$-truss on \tsse]{\includegraphics[width=0.56\linewidth]{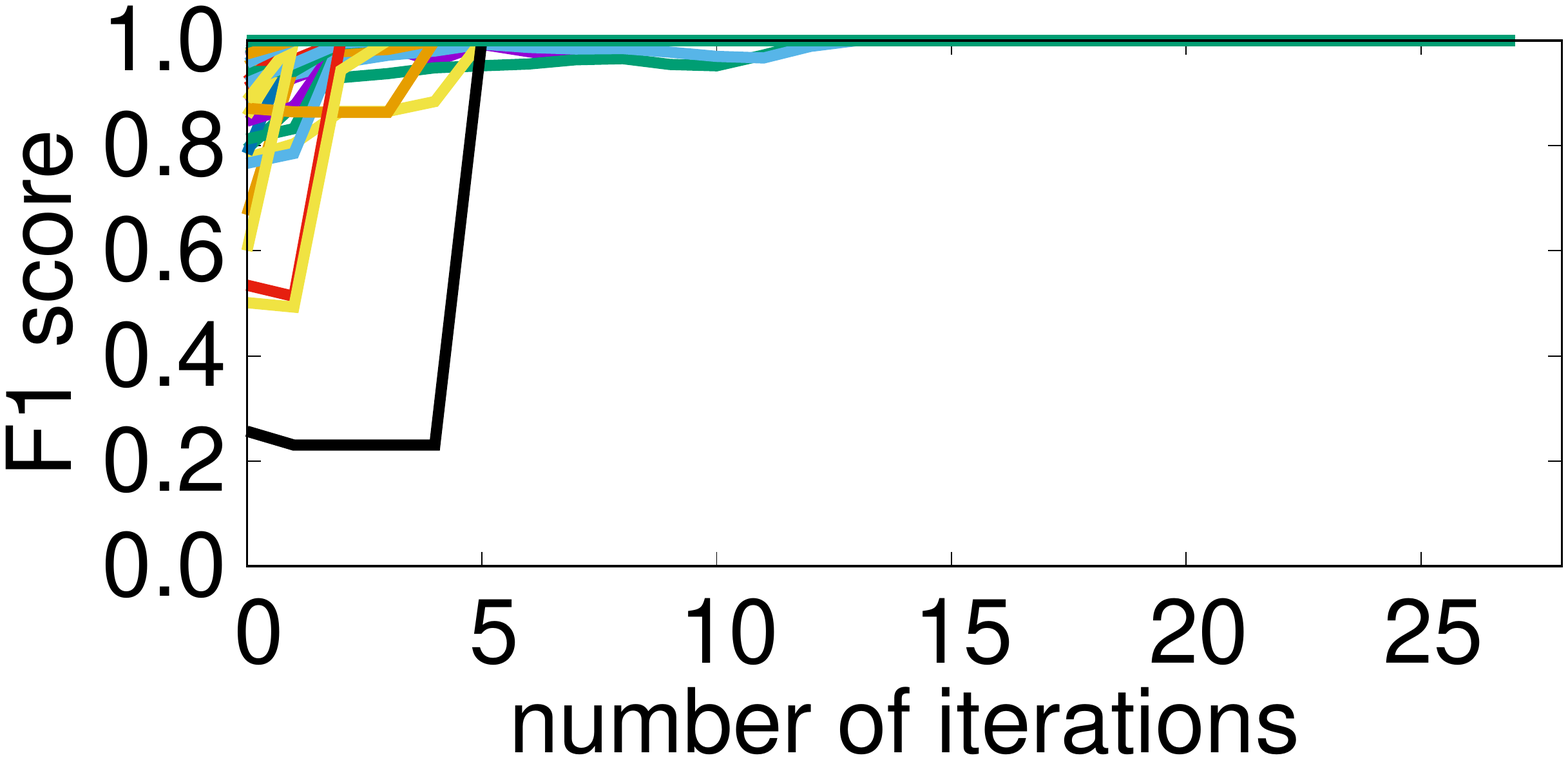}}
\vspace{-2ex}
\caption{\hl{\textit {Evolution of densest subgraphs (leaves). Each line shows the evolution of a leaf. We limit to subgraphs with at least 10 vertices to filter out the trivial ones. Almost all leaves are captured in the first few iterations.}}}
\vspace{1ex}
\label{fig:leafEvol}
\end{figure}

\subsubsection{Evolution of the densest regions}
In the hierarchy of dense subgraphs computed by our algorithms, the leaves  are the most important part (see \cref{fig:toys} and \cref{fig:toynucleus}), since 
those subgraphs have the highest edge density $(|E|/{|V|\choose 2})$, pointing to the significant regions.
Note that the $r$-cliques in a leaf subgraph have same $\kappa$ values and they are the local maximals, i.e., have greater-or-equal $\kappa$ value than all their neighbors.
For this reason, we monitor how the nodes/edges in the leaf subgraphs form their max-cores/max-trusses during the convergence process.
In the $k$-core decomposition; for a given leaf subgraph $L$, we find the max-cores, $M^i_v$, of all $v \in L$ at each iteration $i$ with respect to the $\tau_i$ values.
Then we measure the F1 score between each $M^i_v$ and $L$, and report the averages for each leaf $L$ in iteration $i$, i.e., $\sum_{v \in L} {M^i_v}/|L|$. We follow a similar way for the $k$-truss case; find the max-trusses for all edges in each leaf, track their F1 scores during convergence with respect to the leaf, and report the averages.

\begin{figure*}[!b]
\centering
\vspace{-7ex}
\captionsetup[subfigure]{captionskip=-0.1ex}
\subfloat[\assk]{\includegraphics[width=3.53cm]{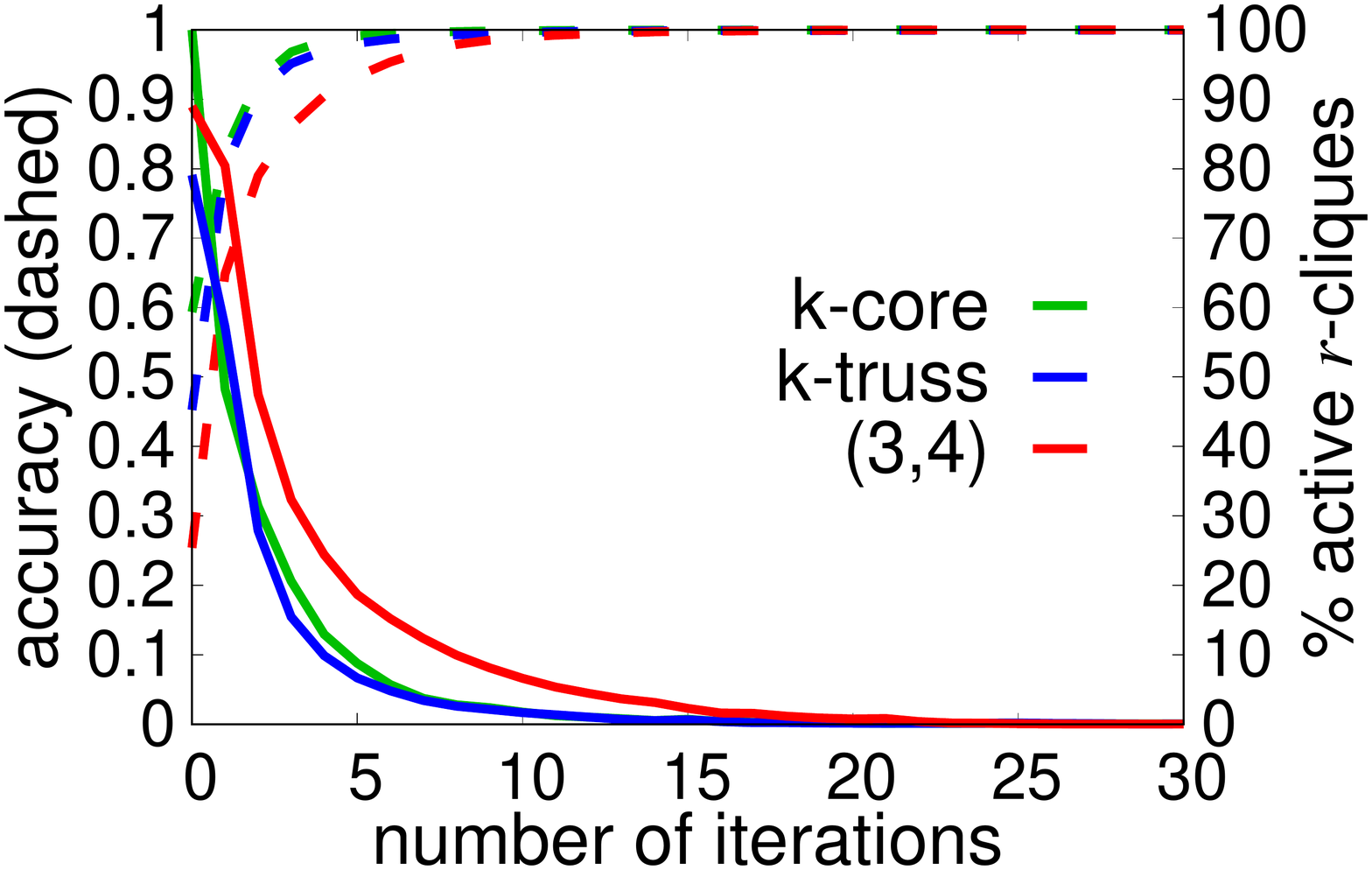}}
\hfill
\subfloat[\slj]{\includegraphics[width=3.53cm]{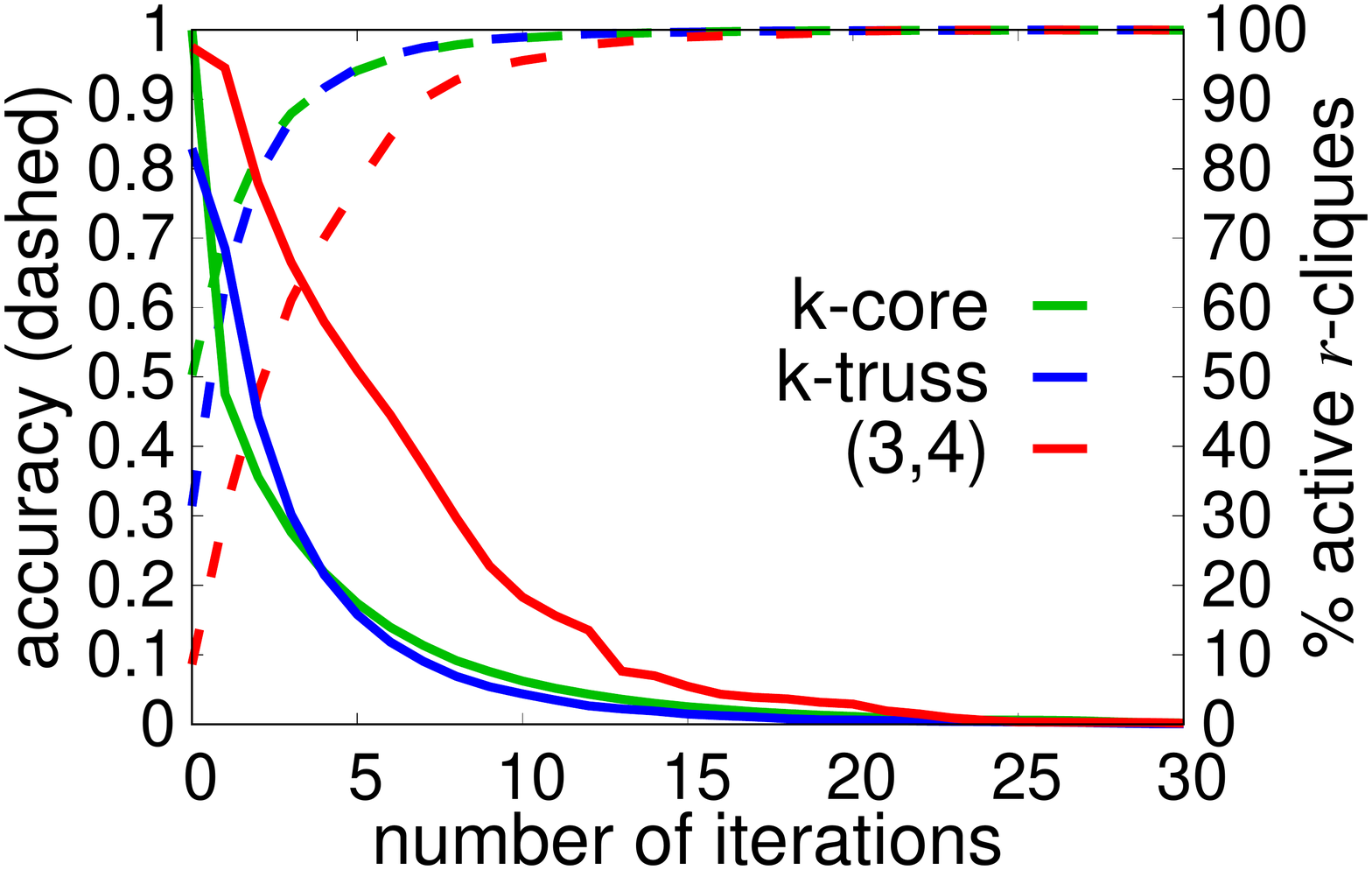}}
\hfill
\subfloat[\so]{\includegraphics[width=3.53cm]{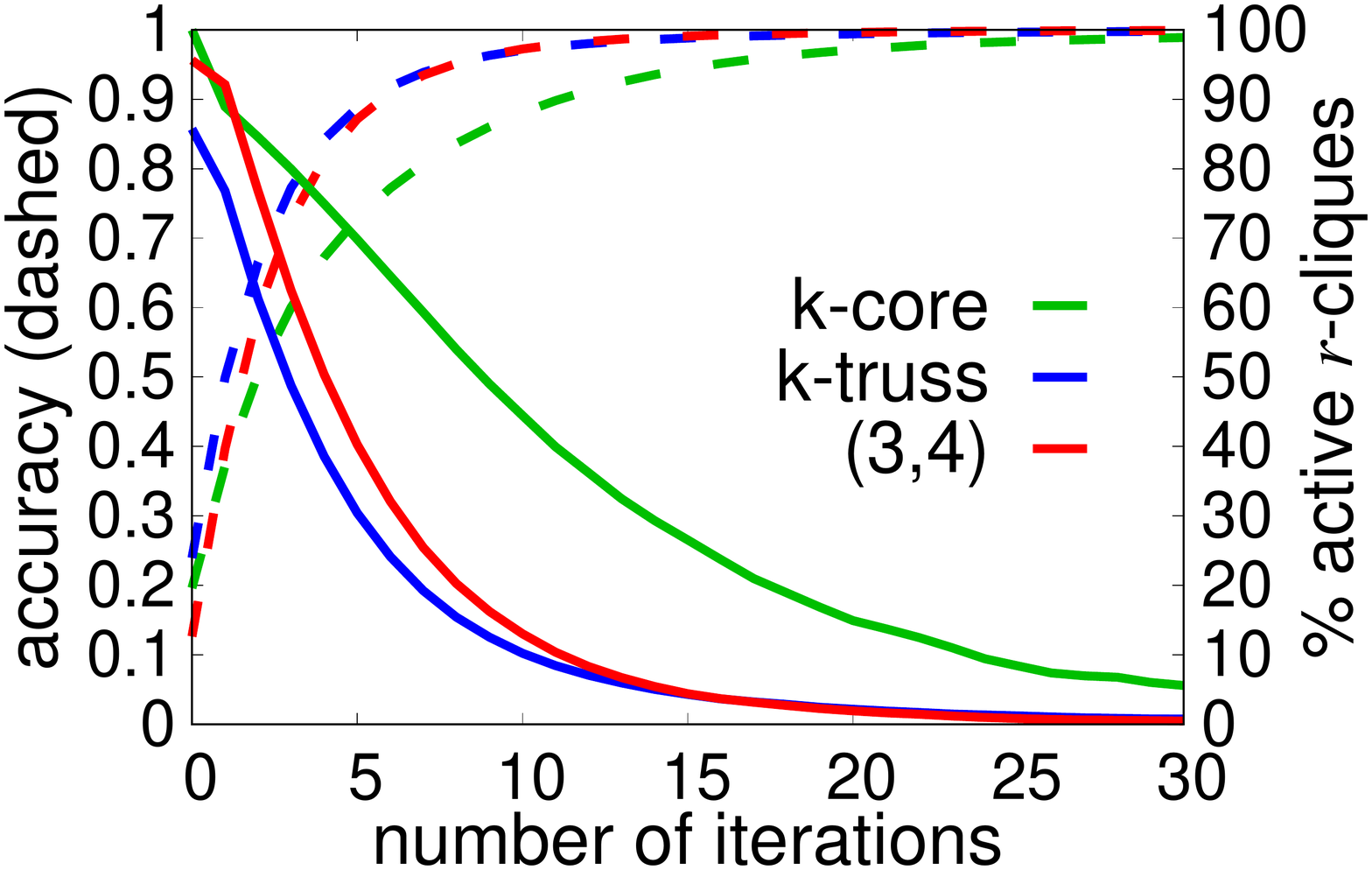}}
\hfill
\subfloat[\wg]{\includegraphics[width=3.53cm]{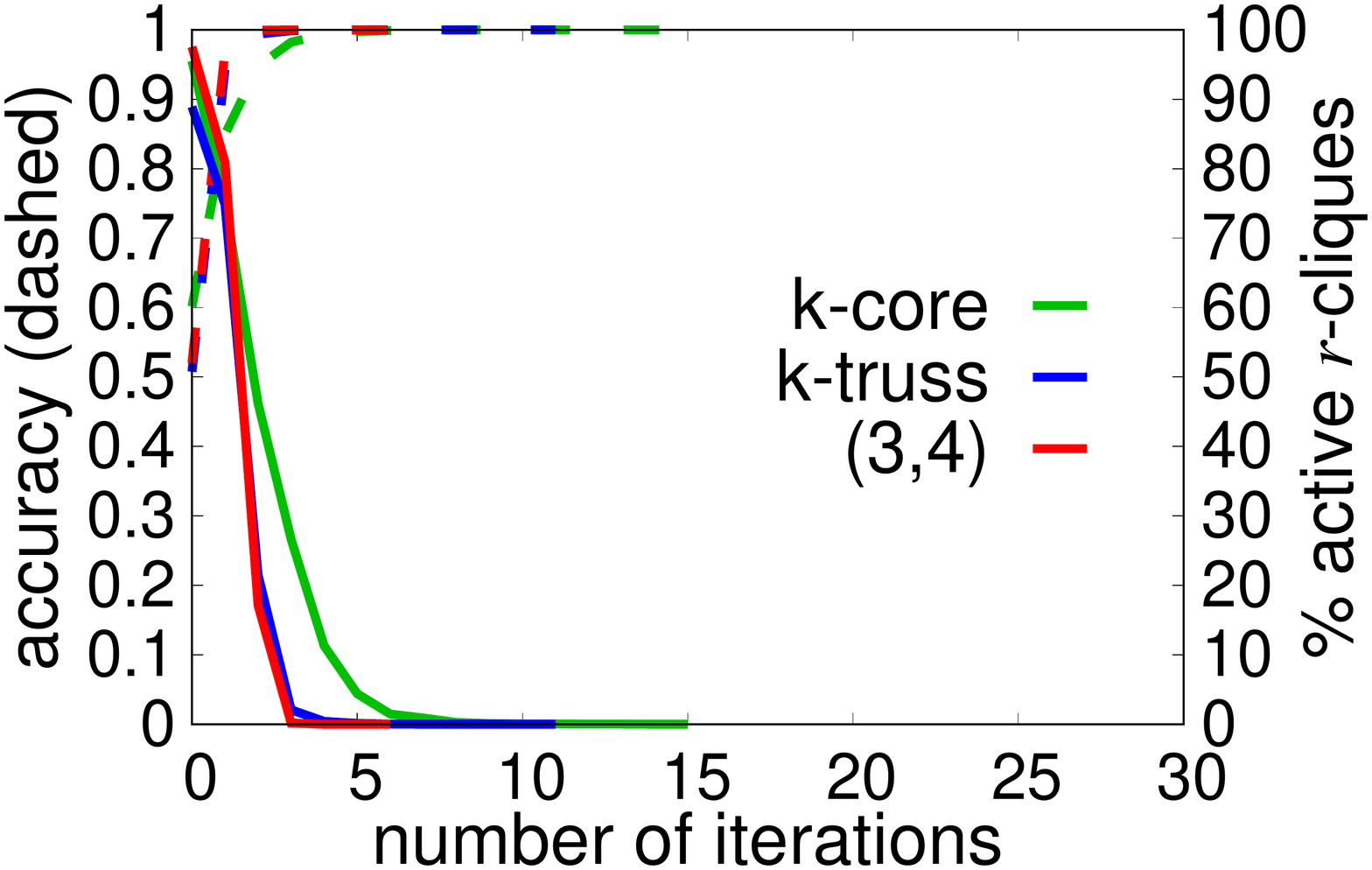}}
\hfill
\subfloat[\wiki]{\includegraphics[width=3.53cm]{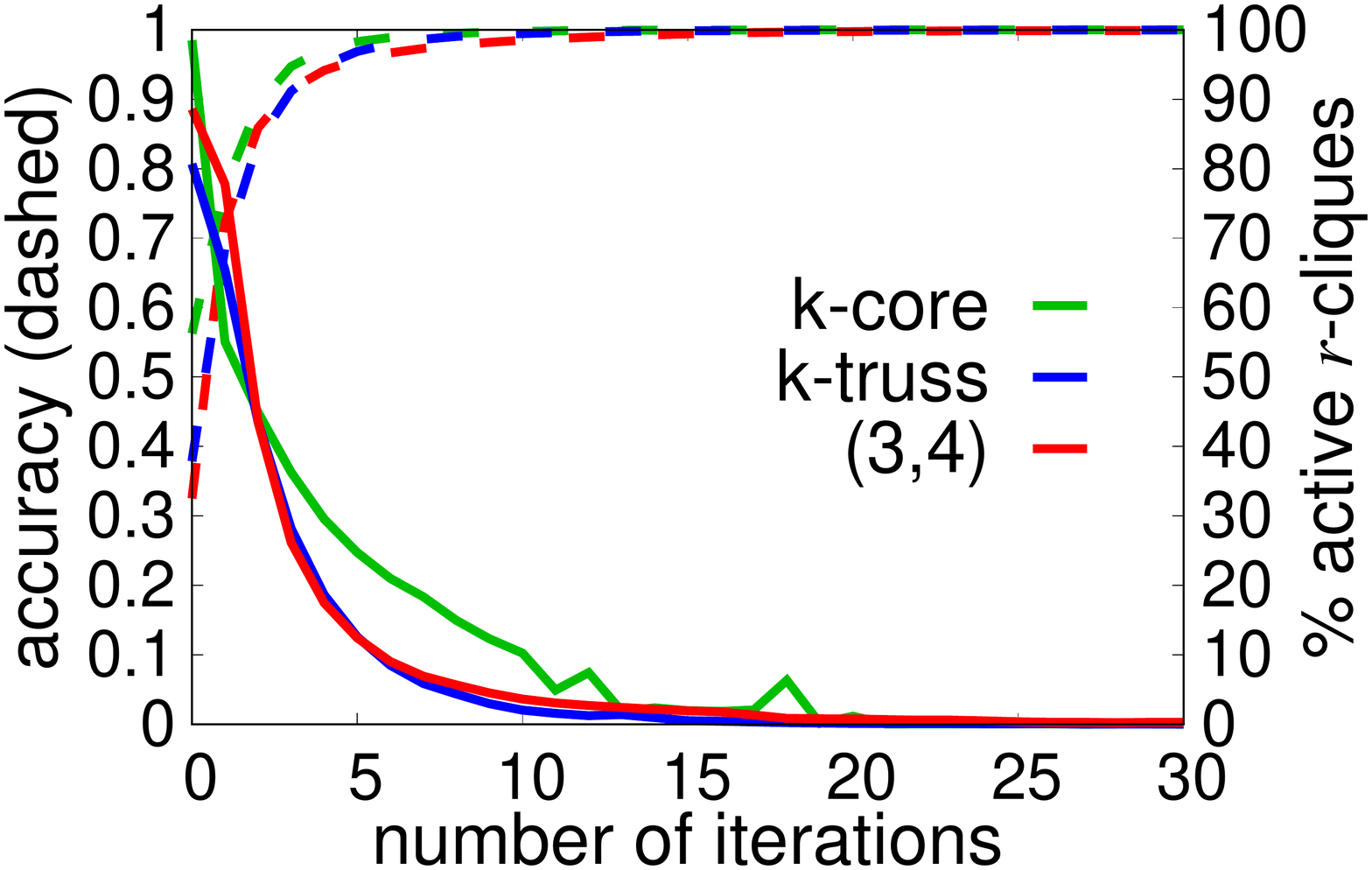}}
\vspace{-1ex}
\caption{\textit {Changes in the ratio of active {$r$}-cliques and the accuracy of {$\tau$} indices during the computation, 
When the ratio of active {$r$}-cliques goes below {$40$}\%, $\tau$ indices provide {$91\%, 89$}\% and {$92$}\% accurate results on average for {$k$}-core, {$k$}-truss and {$(3,4)$} nucleus decompositions, respectively.
If the ratio is below {$10$}\%, more than {$98$\%} accuracy is achieved in all decompositions.
}}
\label{fig:dist}
\vspace{0ex}
\end{figure*}

\cref{fig:leafEvol} presents the results for a representative set of graphs (similar trends observed for other graphs).
Each line shows the evolution of a leaf subgraph during the convergence process and we only considered the subgraphs with at least 10 vertices to filter out the trivial ones.
We observe that almost all leaves are captured in the first few iterations.
Regarding the \fb network, it has 3 leaves in the $k$-core case and all can be found in 7 iterations, and 5 iterations are enough to identify 28 of 33 leaves in $k$-truss decomposition.
This trend is also observed for the other graphs; 5 iterations find 78 of the 85 leaves in $k$-core decomposition of \wn, and 39 of the 42 leaves in $k$-truss decomposition of \sse.

\subsubsection{Predicting Convergence}
Number of iterations for convergence depends on the graph structure, as shown in \cref{tab:converge}.
We cannot know whether a particular $r$-clique has converged by tracking the stability of its $\tau$ index since there can be temporary plateaus (see \cref{fig:plato}).
However, we know which $r$-cliques are active or idle in each iteration thanks to the notification mechanism in \AND algorithm.
We propose using the  ratio of active $r$-cliques as an indicator.

We examine  the relation between the ratio of active  $r$-cliques  and the accuracy ratios of the $r$-cliques.
\cref{fig:dist} presents the results for a set of graphs on all decompositions.
We observe that when the ratio of active $r$-cliques goes below $40$\% during the computation, $91$\%, $89$\%, and $92$\% accurate results are obtained for $k$-core, $k$-truss, and $(3,4)$ nucleus decompositions, on average.
When the ratio goes below $10$\%, over  $98$\% accuracy is achieved in all decompositions.
The results show the ratio of active $r$-cliques is a helpful guide to find almost-exact results can be obtained faster.
Watching for $10$ or $40$\% of active $r$-cliques yields nice trade-offs between runtime and quality.
Watching the $40$\% threshold provides $3.67$, $4.71$, and $4.98$ speedups with respect to full computation in $k$-core, $k$-truss, and $(3,4)$ nucleus decompositions, respectively, and the speedups for $10$\% threshold are $2.26$, $2.81$, and $3.25$ (more details in \cref{sec:tradeoff}).

\subsection{Runtime performance}\label{sec:runtime}
\noindent We evaluate the performance of our algorithms and seek to answer the following questions:\\

\begin{compactitem}[\leftmargin=0.1ex $\bullet$]
\item What is the impact of the notification mechanism (in \cref{sec:notif}) on \AND algorithm?
\item How does the \AND algorithm scale with more threads? How does it compare to  sequential peeling?
\item What speedups are achieved when a certain amount of accuracy is sacrificed?
\end{compactitem}

\subsubsection{Impact of the notification mechanism}
We check the impact of the notification mechanism for $k$-truss and $(3,4)$ cases. We use 24 threads and \cref{tab:nm} presents the results where \AND is the entire \cref{alg:AND} with notification mechanism and \AND-nn does not have the notifications -- missing the orange lines in~\cref{alg:AND}.
We observe that \cref{alg:AND} brings great improvements, reaching up to $3.98$ and $3.16$ over speedups over \AND-nn for $k$-truss and $(3,4)$ cases. 
We use \AND algorithm (with notification mechanism) in the rest of the experiments.

\begin{table}[!t]
\centering
\small
\vspace{1ex}
\renewcommand{\tabcolsep}{4pt}
\caption{\textit{ Impact of the notification mechanism. \AND-nn does not use notifications. Using 24 threads, notification mechanism yields speedups  up to $3.98$ and $3.16$ for $k$-truss and $(3,4)$ cases.}}
\label{tab:nm}
\begin{tabular}{|c|r|r|r|r|r|r|r|r|}\hline
(seconds) & \multicolumn{3}{c|}{$k$-truss} & \multicolumn{3}{c|}{$(3,4)$}\\ 
 Graphs & \AND-nn & \AND & \textit{Speedup} & \AND-nn & \AND & \textit{Speedup} \\ \hline
\tfb	&$	0.45	$&$	0.35	$&$	1.29	$	&$	34.4	$&$	22.2	$&$	1.55	$	\\	\hline
\ttw	&$	3.89	$&$	2.23	$&$	1.74	$	&$	178.7	$&$	59.6	$&$	3.00	$	\\	\hline
\tsse	&$	2.50	$&$	1.46	$&$	1.72	$	&$	105.5	$&$	49.6	$&$	2.13	$	\\	\hline
\twg	&$	3.15	$&$	1.25	$&$	2.52	$	&$	25.7	$&$	16.9	$&$	1.53	$	\\	\hline
\twn	&$	2.38	$&$	0.60	$&$	3.98	$	&$	220.5	$&$	69.8	$&$	3.16	$	\\	\hline
\end{tabular}
\vspace{2ex}
\end{table}

\begin{figure}[!t]
\vspace{-0.5ex}
\centering
\captionsetup[subfigure]{captionskip=-0.5ex}
\hspace{-24ex}
\subfloat[{$k$}{-core}]{\includegraphics[width=0.52\linewidth]{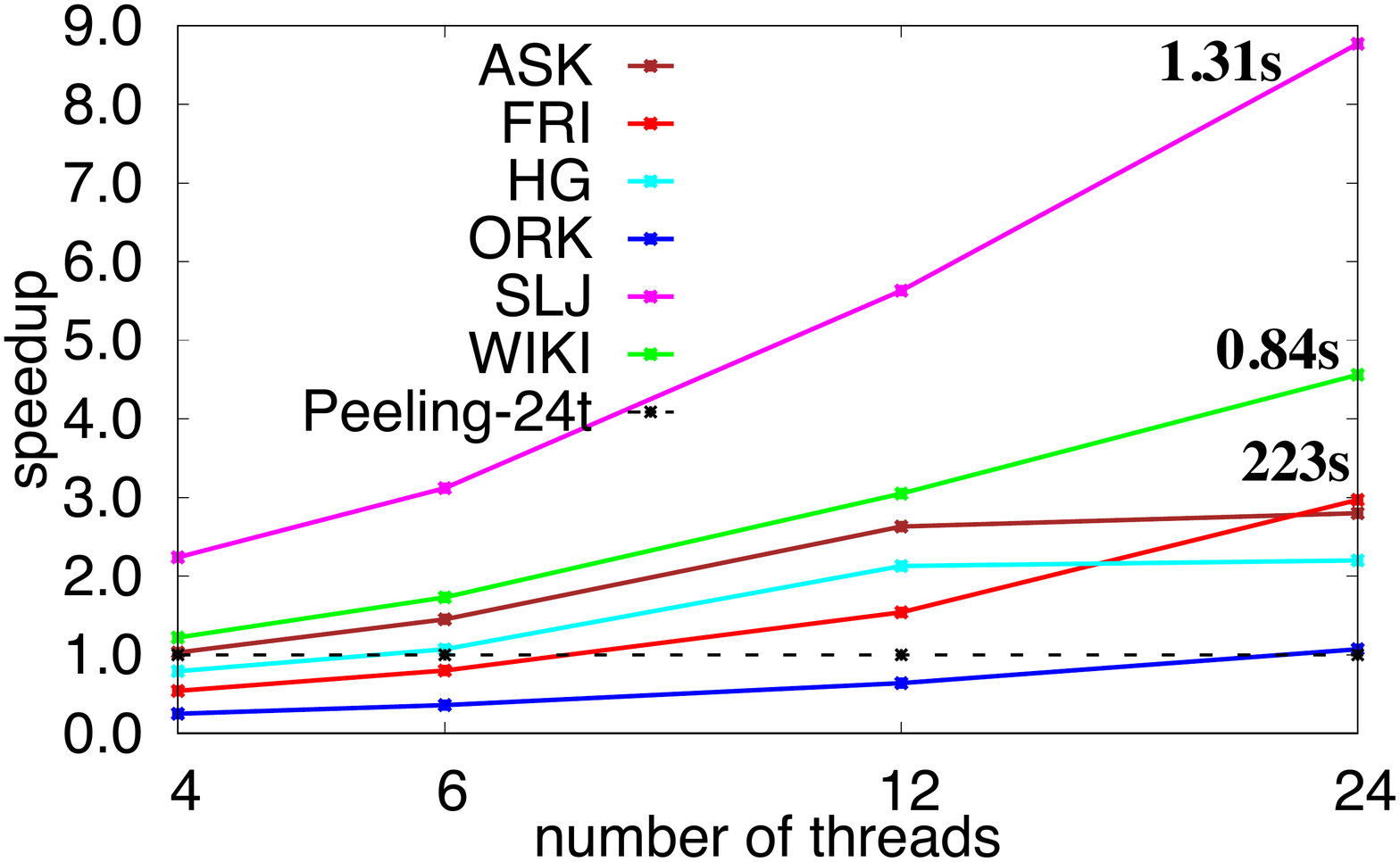}}
\subfloat[{$(3,4)$}]{\includegraphics[width=0.52\linewidth]{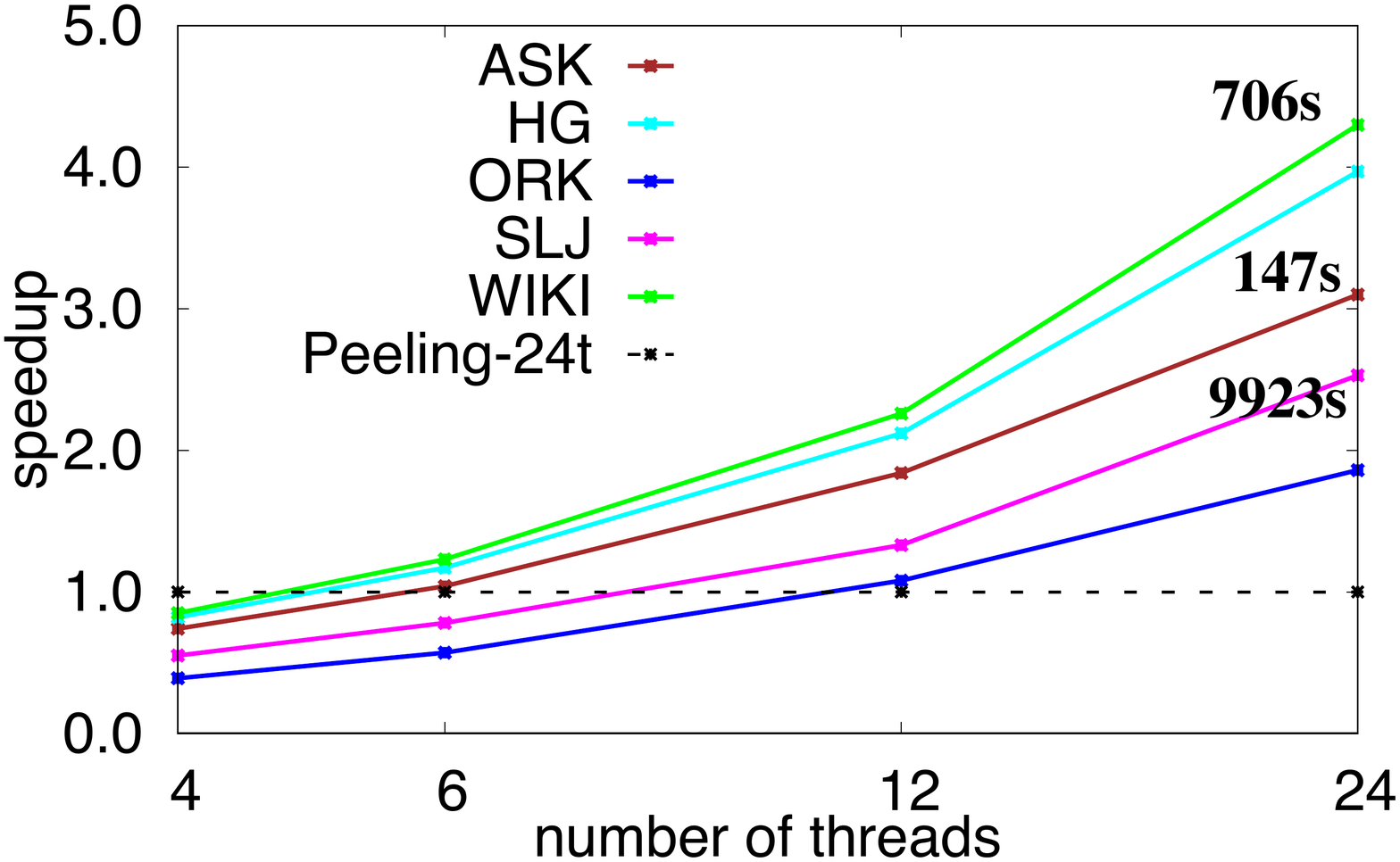}}
\hspace{-24ex}
\vspace{-1ex}
\caption{\textit{Speedups of the parallel computations with respect to the peeling computations (with 24 threads) for the {$k$}-core and {$(3,4)$} nucleus decompositions.
We used {$4, 6, 12$}, and {$24$} threads where each thread is assigned to a single core.
On average, $k$-core computations are performed $3.83$x faster when the number of threads goes from 4 to 24. This increase is $4.7$x for $(3,4)$ case.
Runtimes with 24 threads are annotated for some graphs.
Speedup numbers increase with more threads and faster solutions are possible with  more cores.}}
\vspace{1ex}
\label{fig:scale}
\end{figure}

\subsubsection{Scalability and comparison with peeling}
Given the benefit of notification mechanism, we now compare the runtime performances of \AND (\cref{alg:AND}) and the peeling process (\cref{alg:peeling}) on three decompositions.
Our machine has $24$ cores in total ($12$ in each socket) and we perform full computations until convergence with 4, 6, 12, and 24 threads.
Note that our implementations for the baseline peeling algorithms are efficient;
for instance~\cite{WaCh12} computes the truss decomposition of \assk graph in 281
secs where we can do it in 74 secs, without any parallelization. In addition, for \so and \slj graphs, we compute the truss 
decompositions in 352 and 81 secs whereas~\cite{Huang14} needs 2291 and 1176 secs (testbeds in~\cite{WaCh12} and~\cite{Huang14} are similar to ours).
For the $k$-truss and $(3,4)$ nucleus decompositions, triangle counts per edge and four-clique counts per triangle need to be computed and we parallelize these parts for both the peeling algorithms and \AND, for a fair comparison. Rest of the peeling computation is sequential.
\cref{fig:scale} and~\cref{fig:intro} present the speedups by \AND algorithm over the (partially parallel) peeling computation \textbf{with 24 threads} on $k$-core, $k$-truss, and $(3,4)$ nucleus decompositions.
For all, \AND with 24 threads obtains significant speedups over the peeling computation.
In particular, with 24 threads \AND is $8.77$x faster for the $k$-core case on the \slj, $6.3$x faster for the $k$-truss decomposition on \assk, and $4.3$x faster for the $(3,4)$ nucleus case on \wiki graph.
In addition, our speedup numbers increase by more threads.
On average, $k$-core computations are performed $3.83$x faster when the number of threads are increased from 4 to 24. This increase is $4.8$x and $4.7$x for $k$-truss and $(3,4)$ cases.
Our speedup numbers increase with more threads and \emph{faster solutions are possible with more cores}.

\hl{
\noindent \emph{Recent results:} There is a couple recent studies, concurrent to our work, that introduced new efficient parallel algorithms for $k$-core~\cite{Dhulipala17} and $k$-truss~\cite{Shaden, Pingali, Kamesh} decompositions. Dhulipala et al.~\cite{Dhulipala17} have a new parallel bucket data structure for $k$-core decomposition that enables work-efficient parallelism, which is not possible with our algorithms. They present speedups to $23.6$x on \fri graph with 72 threads. Regarding the $k$-truss decomposition, the HPEC challenge~\cite{hpec} attracted interesting studies that parallelize the computation~\cite{Shaden, Pingali, Kamesh}. In particular, Shaden et al.~\cite{Shaden} reports competitive results with respect to the earlier version of our work~\cite{localnucleus}. Note that \textbf{our main focus in this work is a generic framework that enables local computation for $k$-core, $k$-truss, and $(3,4)$ nucleus decompositions}, which has not been discussed in the previous works. Although our algorithms are not work-efficient and more specialized solutions can give better speedups, our local algorithms are more generally applicable, enable trade-offs between runtime and accuracy, and also enable query-driven scenarios that can be used to analyze smaller subgraphs.
}

\begin{figure}[!t]
\vspace{-2ex}
\centering
\captionsetup[subfigure]{captionskip=-1.5ex}
\hspace{-24ex}
\subfloat[{$k$}{-truss}]{\includegraphics[width=0.54\linewidth]{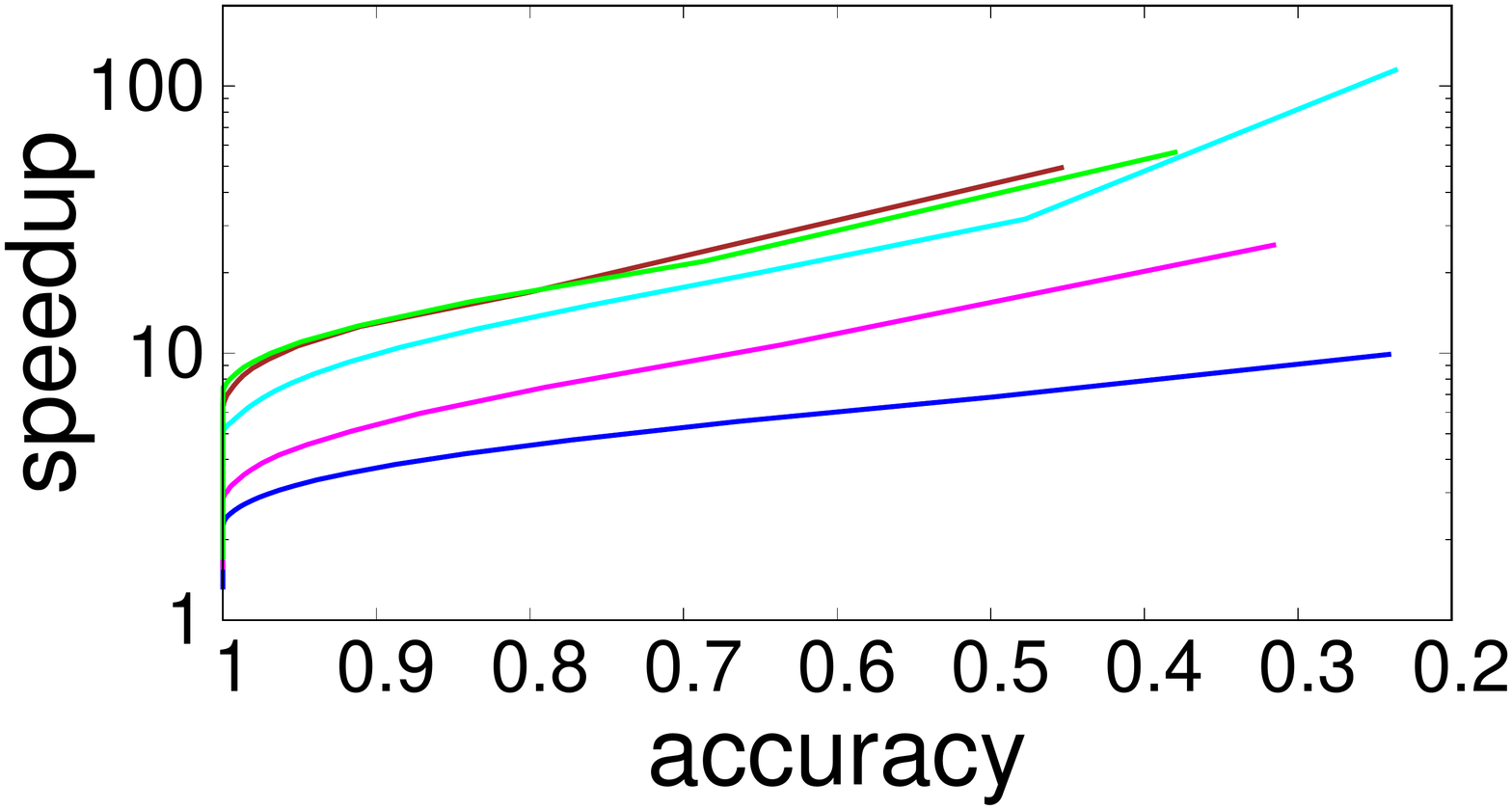}}
\hspace{-2ex}
\subfloat[{$(3,4)$}{-nucleus}]{\includegraphics[width=0.54\linewidth]{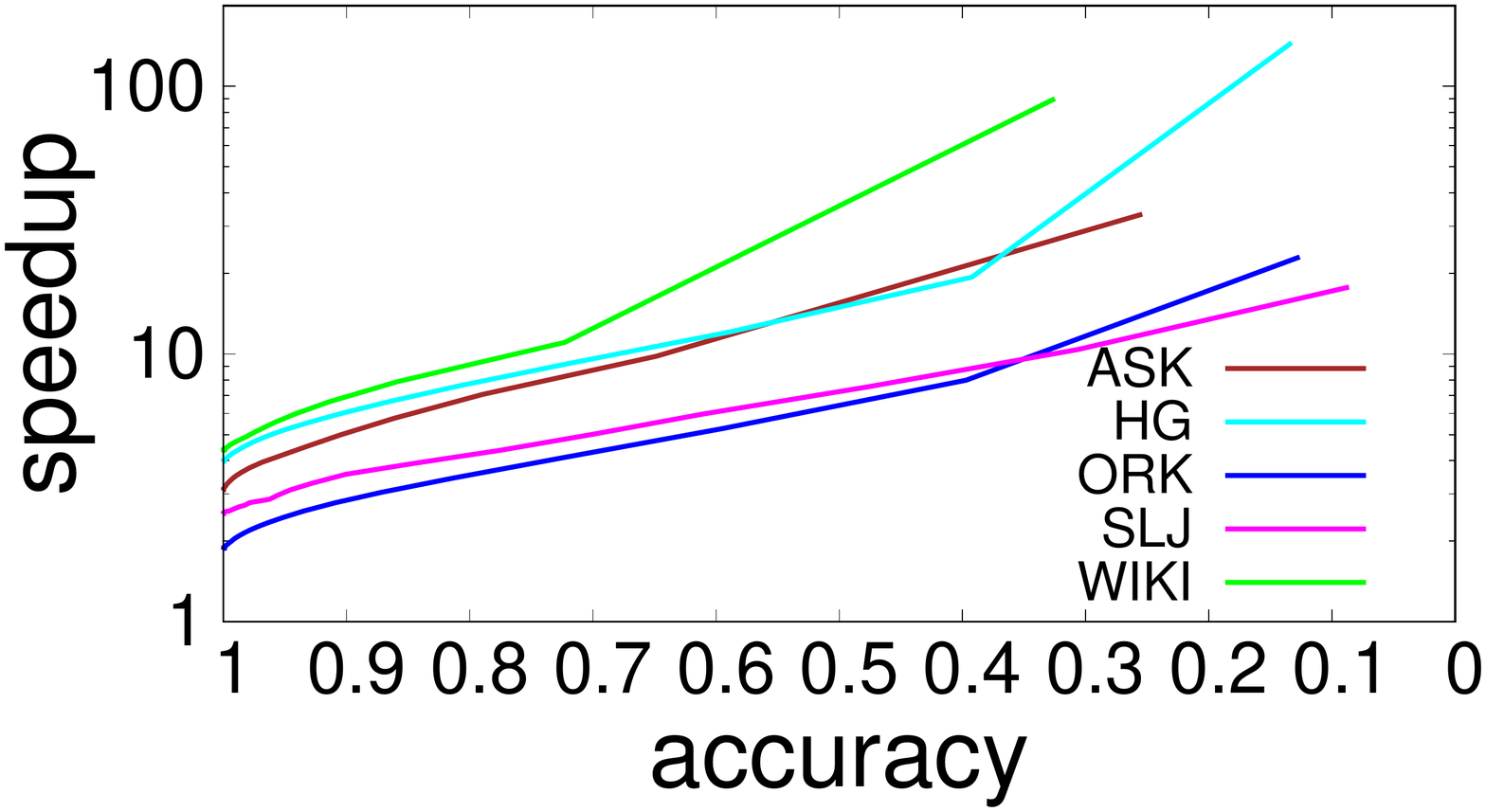}}
\hspace{-24ex}
\vspace{0ex}
\caption{\textit{Runtime/accuracy tradeoff. We show the potential for speedups with respect to the peeling computations for the {$k$}-truss and {$(3,4)$} nucleus decompositions. Speedups at full accuracy correspond to the speedups with 24 threads in \cref{fig:scale}. Number of iterations (and accuracy) decrease on the x-axis.
We reach up to $15$x and $9$x speedups on {$k$}-truss and {$(3,4)$} cases when 0.8 accuracy is allowed.
}}
\vspace{1ex}
\label{fig:spcor}
\end{figure}

\subsubsection{Runtime and accuracy trade-off}\label{sec:tradeoff}

We check the speedups for the approximate decompositions in the intermediate steps during the convergence.
We show how the speedups (with respect to peeling algorithm with 24 threads) change when a certain amount of accuracy  in $\kval_s$ indices is sacrificed.
Figure~\ref{fig:spcor} presents the behavior for $k$-truss and $(3,4)$ nucleus decompositions on some representative graphs.
We observe that speedups for the $k$-truss decomposition can reach up to $15$x when 0.8 accuracy is allowed. For $(3,4)$ nucleus decomposition, up to $9$x speedups are observed for the same accuracy score.
Overall, our local algorithms are enable to enjoy different trade-offs between the runtime and accuracy.

\hl{
\subsection{{\large \bf \PAND} to estimate \boldmath{$\kval_2$} and \boldmath{$\kval_3$} values}\label{sec:app}

\noindent
So far, we have studied the performance of our algorithms on the full graph.  Now, we will look at how we can apply  similar ideas to a  portion of the graph using the 
\PAND algorithm described  at end of \cref{sec:notif}.  We will apply \PAND to  the ego networks and show that it can be used to estimate $\kval_2$ values (core number) of vertices and $\kval_3$ values (truss number) of edges.
Ego network of a vertex $u$ is defined as the induced subgraph among $u$ and its neighbors.
It has been shown that ego networks in real-world networks exhibit low conductance~\cite{Gleich12} and also can
be used for friend suggestion in online social networks~\cite{Epasto15Ego}.
Accurate and fast estimation of core numbers~\cite{OBrien14} is important in the context of network experiments (A/B 
testing)~\cite{Ugander13} where a random subset of vertices are exposed to a treatment and responses are analyzed to measure the impact of a new feature in online social networks.

For the core number estimation of a vertex $u$, we apply \PAND on $u$ and its neighbor vertices, i.e., $u \cup \cN_2(u)$, and report $\kappa_2(u)$.
Indeed, the application of \PAND on ego-network for core number estimation is the same as the \textit{propagating estimator} in~\cite{OBrien14} for distance-1 neighborhood.
Here we generalize the same concept to estimate truss numbers of edges.
Regarding the truss number estimations, we define the ego network of an edge $e$ as the set of neighbor edges that participates in a common triangle ($\cN_3(e)$). Thus, we apply \PAND on $e \cup \cN_3(e)$ and report $\kappa_3(e)$ as a truss number estimate.
Figures~\ref{fig:12_kappa} and \ref{fig:23_kappa} present the results for core and truss number estimations.
We selected vertices/edges with varying core/truss numbers (on the x-axis) and check  the accuracy of  \PAND estimations.
Ground-truth values are shown with green lines (note that y-axes in \cref{fig:12_kappa} are log-scale).
We also show the degrees/triangle counts of the vertices/edges in red as a baseline.
Overall, \PAND yields almost exact estimations for a wide range of core/truss numbers.
On the other hand,  degree of a vertex gives a close approximation to the core number for smaller degrees, but it fails for large values.
This trend is similar for the truss numbers and triangle counts.

\begin{figure*}[!t]
\vspace{-15ex}
\centering
\captionsetup[subfigure]{captionskip=-11ex}
\subfloat[Core number estimations for \twg and \tassk]
{
\includegraphics[width=0.45\linewidth]{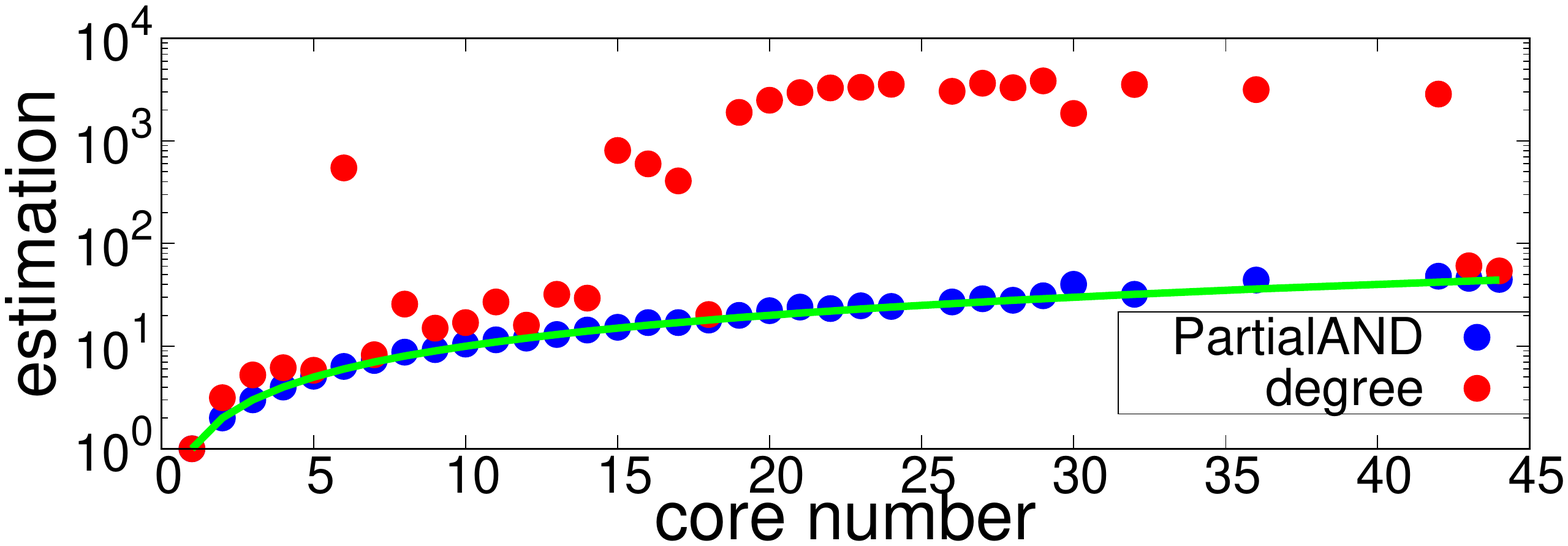}
\includegraphics[width=0.45\linewidth]{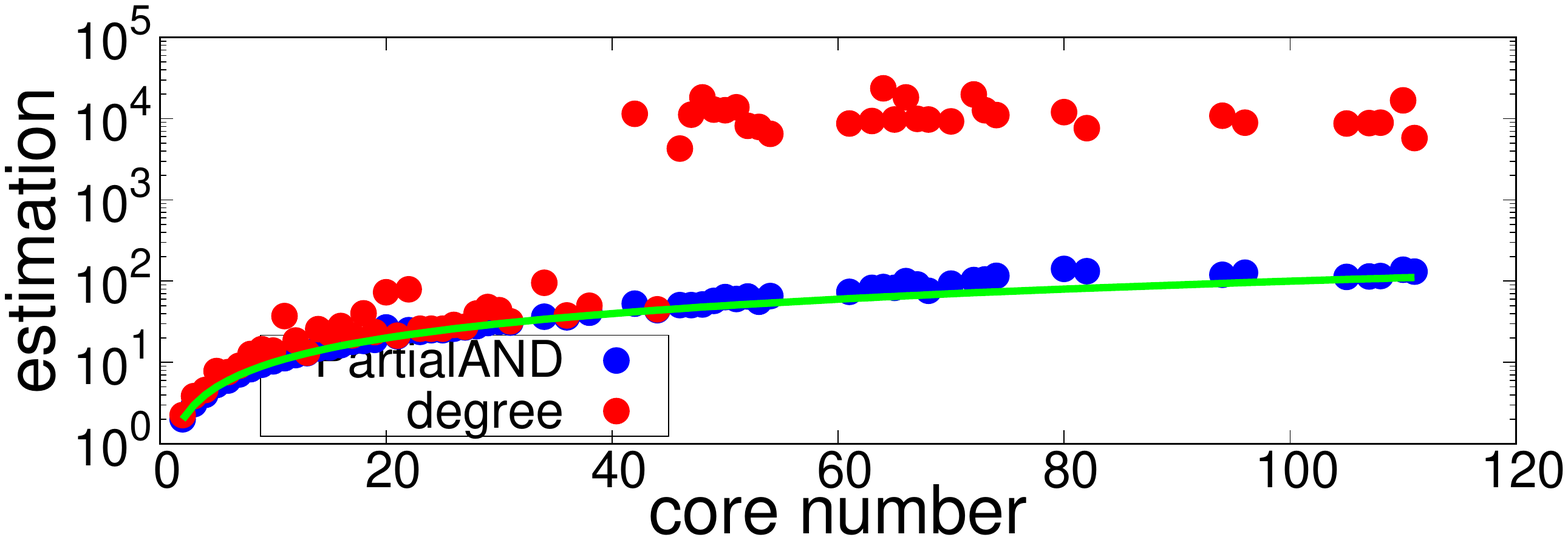}
\label{fig:12_kappa}}
\vspace{-24ex}
\subfloat[Truss number estimations for \ttw and \tsth]
{\hspace{-2ex}
\includegraphics[width=0.45\linewidth]{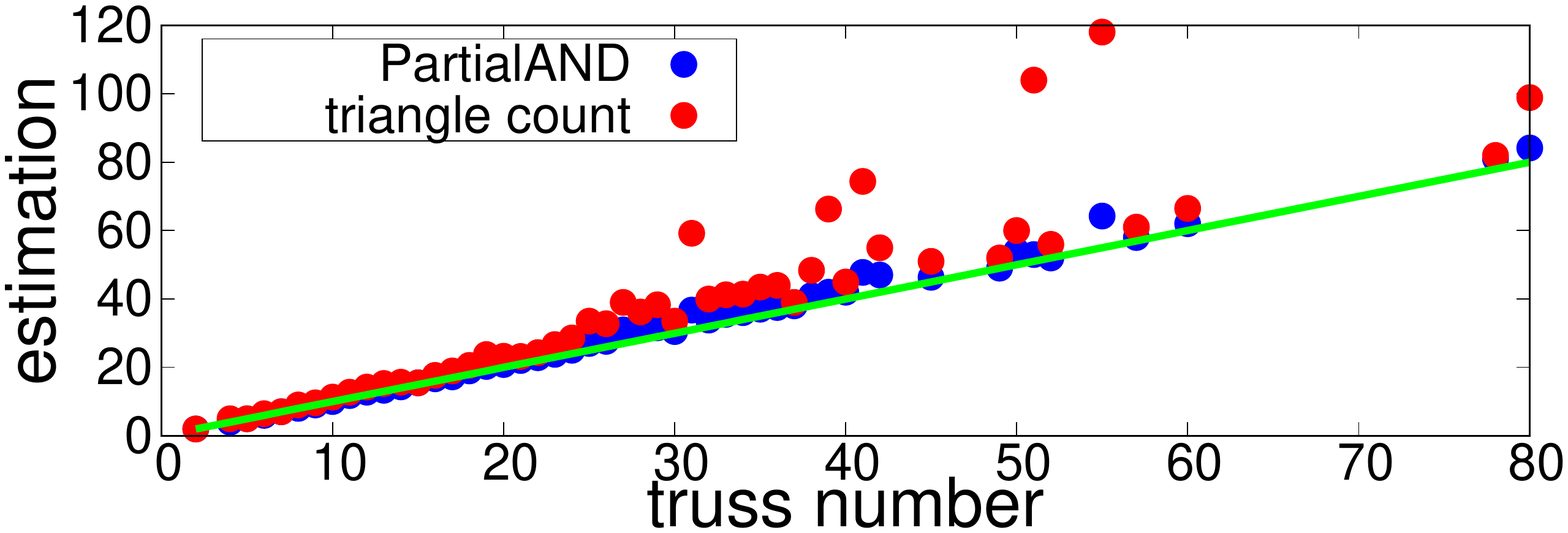}
\includegraphics[width=0.45\linewidth]{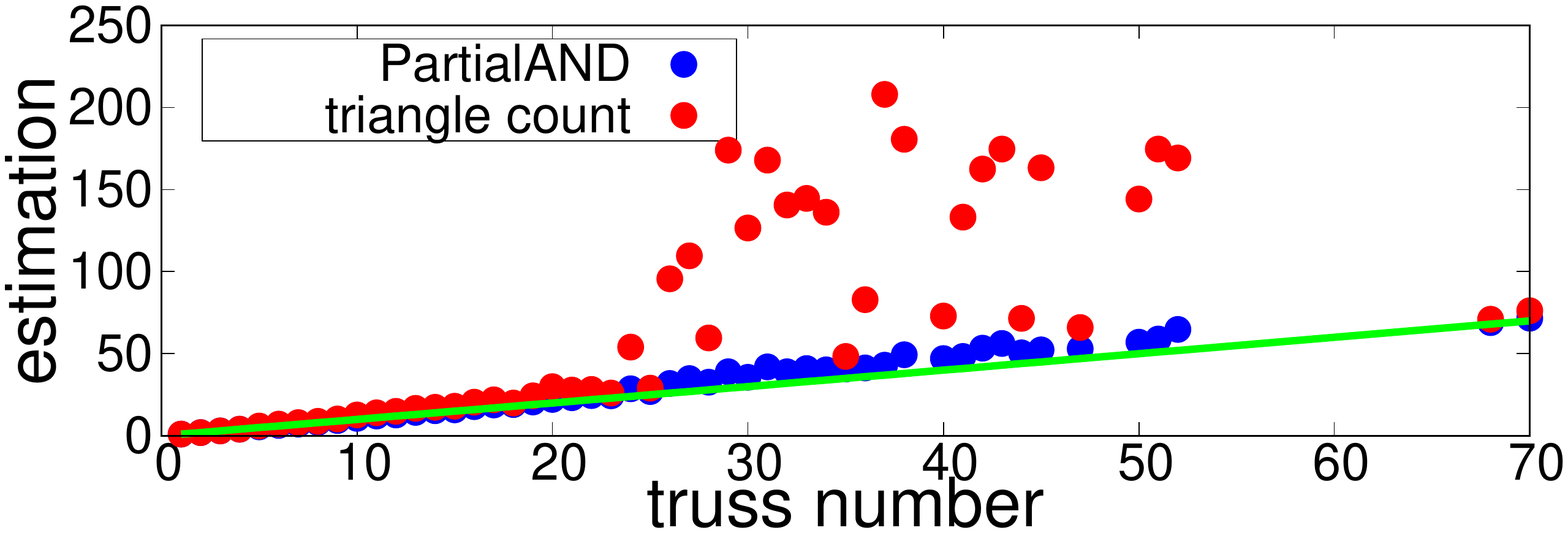}
\label{fig:23_kappa}}
\vspace{-10ex}
\caption{\textit{Accuracy of core and truss number estimations by \PAND. Top two charts present the core number estimations by \PAND and degree with respect to the ground-truth (green line) on \wg and \assk. Bottom two present the truss number estimations for \tw and \sth (green line is the ground-truth). \PAND estimates the core and truss numbers accurately for a wide range. Results for other graphs are similar and omitted for brevity.}}
\vspace{-4ex}
\label{fig:ego_kappa}
\end{figure*}

Regarding the runtime, \PAND on ego networks only takes a fraction of a second -- way more efficient than computing the entire core/truss decomposition. For instance, it takes only $0.23$ secs on average to estimate the core number of any vertex in the \so network whereas the full $k$-core decomposition needs $11.4$ secs.
It is even better in the $k$-truss case; \PAND takes $0.017$ secs on average to estimate a truss number of an edge in \sth network where the full $k$-truss computation takes $73$ secs.

}

\vspace{-2ex}
\section{Related Work}

\noindent Previous attempts to find the approximate core numbers (or $k$-cores) focus on the neighborhood of a vertex within a certain radius~\cite{OBrien14}.
It is reported that if the radius is at least  half of the diameter, close approximations can be obtained.
However, given the small-world nature of the real-world networks, the local graph within a distance of half the diameter is too large to compute.
In our work, we approximate the $k$-core, $k$-truss, and $(r,s)$ nucleus decompositions in a rigorous and efficient way that does not depend on the diameter.

Most related study is done by Lu et al.~\cite{Lu16}, where they show that iterative $h$-index computation on vertices result in the core numbers.
Their experiments on smaller graphs also show that $h$-index computation provides nice trade-offs for time and quality of the solutions.
In our work, we generalized the iterative $h$-index computation approach for \textit{any} nucleus decomposition that subsumes the $k$-core and $k$-truss algorithms.
Furthermore, we give provable upper bounds on the number of iterations for convergence.
Apart from that work, Govindan et al.~\cite{Govindan16} use the iterative $h$-index computation to design space-efficient algorithms for estimating core numbers.
Distributed algorithms in~\cite{Montresor13} and out-of-core approaches in~\cite{Khaouid15, Wen16, Cheng11} also make use of similar ideas, but only for core decomposition.
Montresor et al.~\cite{Montresor13} present a bound for the number of iterations, which is basically $|V|-1$, much looser than ours.

Regarding the parallel computations, Jiang et al.~\cite{Jiang14} introduced parallel algorithms to find the number of iterations needed to find the empty $k$-core in random hypergraphs.
Their work relies on the assumption that the edge density is below a certain threshold and the focus is on the number of iterations only.
Our local algorithms present an alternative formulation for the peeling process, and work for any $k$ value.
For the $k$-truss decomposition, Quick et al.~\cite{Quick12} introduced algorithms for vertex-centric distributed graph processing systems.
 For the same setup, Shao et al.~\cite{Shao14} proposed faster algorithms that can compute $k$-trusses in a distributed graph processing system. Both papers make use of the peeling-based algorithms for computation.
Our focus is on the local computation where the each edge has access to only its neighbors and no global graph information is necessary, thus promise better scalability.

\vspace{-2ex}
\section{Conclusion}
\noindent We introduced a generalization of the iterative $h$-index computations  to  identify  \emph{any} nucleus decomposition and prove convergence bounds.
Our local algorithms are highly parallel and can provide fast approximations to explore time and quality trade-offs.
Experimental evaluation on real-world networks exhibits the efficiency, scalability, and effectiveness of our algorithms for three decompositions.
We believe that our local algorithms will be beneficial for many real-world applications that work in challenging setups. For example, shared-nothing systems can leverage the local computation. 

%

\clearpage
\bibliographystyle{abbrv}
\bibliography{paper}

\end{document}